\documentclass[10pt]{article}
\usepackage{subfigure}
\usepackage{amsthm}
\usepackage{amssymb}
\usepackage{amsmath}
\usepackage{color}
\usepackage{verbatim}
\usepackage{multirow}
\usepackage{stmaryrd} 
\usepackage{relsize} 

\usepackage{amsfonts}
\usepackage{epsfig}
\usepackage{graphicx}
\newtheorem{proposition}{Proposition}

\newcommand{\mbf}{\mathbf}

\newcommand{\Rset}{\mathbb R}

\newcommand{\vect}[1]{\boldsymbol{#1}}

\newcommand{\0}{\boldsymbol 0}

\newcommand{\f}{\mathbf f}
\newcommand{\bA}{\mathbf A}                               
\newcommand{\bmM}{\mathbf M}                               

\newcommand{\bK}{\mathbf K}                               
\newcommand{\bI}{\mathbf I}                               
\newcommand{\bn}{\boldsymbol n}
\newcommand{\bssigma}{\boldsymbol \sigma}
\newcommand{\bsSigma}{\boldsymbol \Sigma}
\newcommand{\bmb}{\mathbf b}
\newcommand{\bmB}{\mathbf B}
\newcommand{\bmt}{\mathbf t}
\newcommand{\bmI}{\mathbf I}

\newcommand{\bmx}{\mathbf x}
\newcommand{\bt}{\boldsymbol t}
\newcommand{\bsmu}{\mathbf u}

\newcommand{\bu}{\boldsymbol u}

\newcommand{\bx}{\boldsymbol x}

\newcommand{\bv}{\boldsymbol v}
\newcommand{\bveps}{\boldsymbol \varepsilon}

\newcommand{\dd}{\mathrm{d}}

\newcommand{\red}[1]{#1}

\newcommand{\mc}[1]{\multicolumn{1}{|c|}{#1}}
\newcommand{\tmI}{\tilde{\mathbf I}}
\newcommand{\A}{ \mathlarger{\mathlarger{\mathlarger{\boldsymbol{\mathsf{A}}}}} }
\newcommand{\argmin}{ \mathrm{argmin}}
\makeatletter
\def\dynscriptsize{\check@mathfonts\fontsize{\sf@size}{\z@}\selectfont}
\makeatother
\def\textunderset#1#2{\leavevmode
  \vtop{\offinterlineskip\halign{%
    \hfil##\hfil\cr\strut#2\cr\noalign{\kern-.3ex}
    \hidewidth\dynscriptsize\strut#1\hidewidth\cr}}}
\newcommand{\tr}{\mathrm{tr}}

\graphicspath{{images/}}

\begin{document}

\title{\textbf{\Large{Non-regularised Inverse Finite Element Analysis for 3D Traction Force Microscopy}}}

\author{Jos\'e J Mu\~noz}


%

\maketitle

\centerline{Dept. Mathematics, Laboratori de C\`alcul Num\`eric (LaC\`aN)}
\centerline{Universitat Polit\`ecnica de Catalunya (UPC), Barcelona, Spain}
\centerline{\tt j.munoz@upc.edu, http://www.lacan.upc.es/munoz}

\vspace{2pc}

\abstract{
The tractions that cells exert on a gel substrate from the observed displacements is an increasingly attractive and valuable information in biomedical experiments. The computation of these tractions requires in general the solution of an inverse problem. Here, we resort to the discretisation with finite elements of  the associated direct variational formulation, and solve the inverse analysis using a least square approach. This strategy requires the minimisation of an error functional, which is usually regularised in order to obtain a stable system of equations with a unique solution. In this paper we show that for many common three-dimensional geometries, meshes and loading conditions, this regularisation is unnecessary. In these cases, the computational cost of the inverse problem becomes equivalent to a direct finite element problem. For the non-regularised functional, we deduce the necessary and sufficient conditions that the dimensions of the interpolated displacement and traction fields must preserve in order to  exactly satisfy or yield a unique solution of the discrete equilibrium equations. We apply the theoretical results to some illustrative examples and to real experimental data. Due to the relevance of the results for biologists and modellers, the article concludes with some practical rules that the finite element discretisation must satisfy.}


\section{Introduction}

The development of computational methods that allow scientists to accurately quantify the forces that cells exert on their surrounding has attracted a large amount of research \cite{butler02,franck11,schwarz02,toyjanova14,brask15}, which can be also found in recent review articles \cite{schwarz15}. These methods are currently being used to elucidate the proteins that control the mechanical response of cells when undergoing embryo morphogenesis, wound closure or cancer growth, to name a few \cite{brugues14,trepat09}. 

Some of the experimental techniques that were originally developed to measure the cellular tractions used  micromachined substrates \cite{galbraith97}, microneedles, or micro-pilars \cite{roure05}. Nowadays, the most popular methodology is to compute the cell tractions from the measured cell velocities and displacements on a polyacrylamide gel substrate. In some cases, this gel is partially covered by a membrane of polydimethylsiloxane (PDMS) that surrounds the cell monolayer in order to control the initial conditions of the cell migration. The idea of indirectly retrieving the cell tractions from the substrate deformations is founded on the seminal work of Harris et al. \cite{harris80}, and was later experimentally implemented in two \cite{dembo96} and three dimensions \cite{dembo99}. These techniques have been experimentally improved by Toyjanova et al. \cite{toyjanova14} in order to increase the accuracy of the measurements. Figure \ref{f:gel} illustrates the set-up considered in the present paper, where the deformation $\bu_0$ at the top surface of an assumed elastic gel is measured, and the traction field $\bt$ obtained indirectly.

\begin{figure}[!htb]
\centerline{\includegraphics[width=0.95\textwidth]{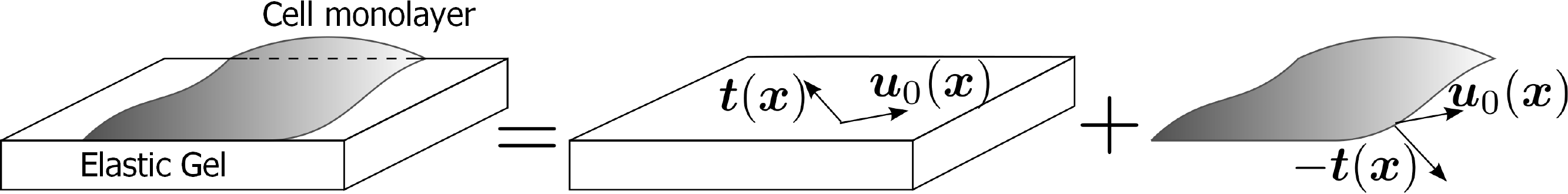}}
\caption{ General set-up in Traction Force Microscopy (TFM). A displacement field $\bu_0$ is imposed by the cell monolayer on the top of an assumed elastic gel is measured, and the traction field $\bt$ computed by solving an inverse elasticity  problem. }
\label{f:gel}
\end{figure}

Computationally, retrieving the tractions $\bt$ exerted by the cells from the measured displacements $\bu_0$ requires the solution of an inverse elasticity problem. In the present paper we analyse the finite element discretisation of this inverse problem. The use of finite elements in inverse analysis is a common practice in scattering problems \cite{beilina14}, localisation of pollutant sources \cite{deng13}, \red{estimation of Robin coefficients \cite{jin10},} or in elasticity problems \cite{ambrosi06,vitale12b}.  So far, the construction of well-posed inverse problems is ensured by resorting to Tikhonov regularisation \cite{ambrosi06,ramm05,schwarz02,vitale12b}, which depends on a penalty parameter. The optimal value of this parameter, which compromises the accuracy of the equilibrium conditions and the condition number of the system of equations has been studied for instance in \cite{hansen00}. We here determine the conditions that give rise to a well-posed discretised inverse elasticity problem in the absence of regularisation. We focus our attention on finite element (FE) discretisations of some commonly employed configurations in Traction Force Microscopy (TFM), also known as Cell Traction Microscopy \cite{vitale12}. We show that the regularisation process is in fact unnecessary, or it can be circumvented by modifying the domain discretisation.

The paper is organised as follows. In Section \ref{s:cont} we present the continuous direct and inverse problems. Section \ref{s:disc} describes the discrete versions of these two problems, and analyses the uniqueness of the solution in the inverse problem according to the dimensions of the discrete traction and displacement fields. Section \ref{s:res} applies the methodology to a toy problem that illustrates the main theoretical results, and to a problem with real experimental data.


\section{Continuous problem in linear elasticity}\label{s:cont}

\subsection{Continuous direct problem}\label{s:cf}

We consider an open connected domain $\Omega\subset\Rset^3$ subjected to homogeneous displacement conditions at a Dirichlet boundary $\Gamma_d\neq\emptyset$ and to surface loads $\bt$ on a Neumann boundary $\Gamma_n$, with $\bar\Gamma_d\cap\bar \Gamma_n=\emptyset$ and with the boundary $\partial\Omega=\overline{\Gamma_d\cup\Gamma_n}$ which is Lipschitz-continuous. The material in $\Omega$ is assumed to obey a linear elastic constitutive law with Lam\'e coefficients $\lambda>0$ and $\mu>0$, which are not necessarily constant in the domain $\Omega$. After neglecting the volumetric forces, the strong form of linear elasticity may be stated as the following boundary value problem \cite{ciarlet88}: 
\begin{eqnarray}
\vect\nabla\cdot\bssigma(\vect u)=\0, &\ \forall \bx\in\mathrm{int}(\Omega),\label{sf1}\\
\bssigma(\bu)\bn=\bt, &\ \forall \bx\in\Gamma_n \label{sf2},\\
\bu=\0, &\ \forall \bx\in\Gamma_d, \label{sf3}
\end{eqnarray}
with $\bssigma(\bu)=\lambda(\nabla\cdot\bu)\mathbf I +\mu(\vect\nabla\bu+\vect\nabla\bu^T)$ denoting the stress tensor. The traction field may contain discontinuities, but it is assumed that $\bt\in T\subseteq L^2(\Gamma_n)$. We point out that the strong form (\ref{sf1})-(\ref{sf3}) and the subsequent results are valid for homogeneous and non-homogeneous problems. \red{Indeed, linear anisotropic or non-homogeneous materials may be handled by resorting to the necessary alternative stress-strain relationships or by just using position  dependent material properties.} In fact, the latter case has motivated the present article. \red{These changes will only affect the computation of the matrices that will be presented in the finite element discretisation, but do not modify the methodology and theoretical results of this article. }

Let us define the following spaces $U$ and $V$,
\begin{eqnarray*}
U=V=(H^1_0(\Omega))^3=\left\{\vect v\in(H^1(\Omega)){}^3 : v_i=0\ \mbox{on}\ \Gamma_d, i=1,2,3\right\},
\end{eqnarray*}
which are endowed with the scalar product $(\vect u,\vect v)=\int_\Omega \vect u\cdot\vect v \mathrm d\Omega $, and equipped with the norm $||\vect v||_\Omega=\left(\sum_{i=1}^3 ||v_i||_{1,\Omega}^2\right)^{1/2}$, where $||v_i||^2_{1,\Omega}=\int_{\Omega} (|v_i|^2+\vect \nabla v_i\cdot\vect\nabla v_i)\mathrm{d}\Omega$. After multiplying by a trial function $\vect v\in V$ the first equation in (\ref{sf1})-(\ref{sf3}), integrating on the domain $\Omega$, integration by parts, and using the boundary conditions, the weak form of problem (\ref{sf1})-(\ref{sf3}) reads \cite{ciarlet88}:
\begin{eqnarray}\label{wf}
\mbox{Find}\ \bu\in U\ \mbox{s.t.}\ a(\bu,\bv)=b(\bv), &\ \forall \bv\in V.
\end{eqnarray}

 The bilinear and linear forms $a(\cdot,\cdot)$ and $b(\cdot)$ are given by,
\begin{eqnarray*}
a(\bu, \bv)&:=&\int_\Omega\bssigma(\bu):\bveps(\bv)\mathrm{d}\Omega, \\
b(\bv)&:=&\int_{\Gamma_n}\bv\cdot\bt\mathrm{d}\Gamma, 
\end{eqnarray*}
where $\bveps(\vect v)=\frac{1}{2}\left(\vect \nabla\bv+(\vect\nabla\bv)^T\right)$ is the small strain tensor.
Since the bilinear form  $a(\cdot,\cdot)$ is continuous and coercive (or $V-$ elliptic) with respect to the space $V$ (see \cite{ciarlet88}, Section 1.2), and we assume that $\Gamma_d\neq\emptyset$, the weak form in (\ref{wf}) accepts only one solution \cite{ciarlet88}. We will denote by $\bu[\bt]$ the solution of problem (\ref{wf}) for a given boundary load $\vect t$. 

In the subsequent paragraphs we will in fact consider a partitioning of the domain $\Omega$ into domains $\Omega_0\subseteq\Omega$ and $\Omega_1=\Omega\backslash\Omega_0$. In sub-domain $\Omega_1$ we assume a \emph{given} displacement field $\bu_1$ that satisfies the elasticity equations,
\begin{eqnarray}
\vect\nabla\cdot\bssigma(\vect u_1)=\0, &\ \forall \bx\in\mathrm{int}(\Omega_1),\nonumber\\
\bssigma(\bu_1)\bn=\bt, &\ \forall \bx\in\Gamma_n\cap\Omega_1 \label{sfp2},\\
\bu_1=\0, &\ \forall \bx\in\Gamma_d\cap\Omega_1. \nonumber
\end{eqnarray}

We will then denote by $\bu[\vect t, \bu_1]$, the solution $\bu_0$ that satisfies the elasticity problem in $\Omega_0$ compatible with $\bu_1$ and the boundary conditions in (\ref{sf2})-(\ref{sf3}),
\begin{eqnarray}
\vect\nabla\cdot\bssigma(\bu_0)=\0, &\ \forall \bx\in\mathrm{int}(\Omega_0),\nonumber\\
\bssigma(\bu_0)\bn=\bt, &\ \forall \bx\in\Gamma_n\cap\Omega_0 \label{sfp1},\\
\bu_0=\0, &\ \forall \bx\in\Gamma_d\cap\Omega_0, \nonumber\\
\bu_0=\bu_1, &\ \forall \bx\in\partial\Omega_0\backslash(\Gamma_n\cup\Gamma_d),\nonumber
\end{eqnarray}

It is important to stress that problem (\ref{sfp1}) aims to find an unknown displacement $\bu_0$, while 
equations (\ref{sfp2}) just give some conditions on the known displacement $\bu_1$. \red{The tractions $\bt$ and displacements $\bu_1$ are known, and $\bu_1$ satisfies the equilibrium equations, so that problem \eqref{sfp2}-\eqref{sfp1} accepts only one solution $(\bu_1, \bu_2)$. This direct problem has no practical interest, but it is used here to ease the presentation of the inverse problem in the next subsection.}

\subsection{Continuous inverse problem}\label{s:ci}

The continuous inverse problem of (\ref{sfp1})-(\ref{sfp2}) consists on assuming instead the knowledge of displacements $\bu_0$, and finding the traction and displacements fields, $\bt$ and $\bu_1$ respectively. Formally, it is stated as,
\begin{eqnarray}\label{eq:if}
\mbox{Given}\  \vect u_0\in U_0(\Omega_0), \mbox{find}\ \bt\in T\subseteq L^2(\Gamma_n) \ \mbox{and}\ \bu_1\in U_1(\Omega_1)\ \mbox{s.t.}\nonumber \\
\phantom{Put it on the right}  \bar b(\bt,\bv)=\bar a(\bu_1,\bv)+\bar c(\bv),\ \forall \bv\in V,
\end{eqnarray}
where the forms $\bar a(\bu_1, \bv)$, $\bar b(\bt,\bv)$ and $\bar c(\bv)$ are given by,
\begin{eqnarray*}
\bar a(\bu_1, \bv)     &:=&\int_{\Omega_1}\bssigma(\bu_1):\bveps(\bv) \mathrm{d}\Omega,\\
\bar b(\bt, \bv)&:=&\int_{\Gamma_n}\bv\cdot\bt\mathrm{d}\Gamma,\\
\bar c(\bv)     &:=&\int_{\Omega_0}\bssigma(\bu_0):\bveps(\bv) \mathrm{d}\Omega.
\end{eqnarray*}

\red{Domain $\Omega_0$ contains the location of the points where $\vect u_0$ is measured. Although it is possible to experimentally measure displacements fields at the interior of tissues or organs, in our examples in Section \ref{s:res}, domain $\Omega_0$ will be limited to the top boundary of the gel, in contact with the cell monolayer (see Figure \ref{f:gel}), while $\Omega_1$ is the interior of the gel. The unknown $\bt$ will correspond in this case to the tractions exerted by the cells on the top of the gel.}

In general, the existence and uniqueness of the solution of (\ref{eq:if}) cannot be guaranteed. This is partially due to the fact that the measured displacements $\bu_0$ may not be a solution of a linear elastic problem, due to the non-linearities of the substrate or to experimental errors. For instance, if $\nabla\cdot\bssigma(\vect u_0)\neq \mathbf 0$ somewhere in $\mathrm{int}(\Omega_0)$, then no traction field satisfying (\ref{eq:if}) can be found. If instead $\nabla\cdot\bssigma(\vect u_0)=\mathbf  0$ everywhere in $\Omega_0$, the choice $\vect t=\bssigma(\vect u_0)\vect n\Big\vert_{\Gamma_n\cap\Omega_0}$ and $\bu_1$ the solution of the elasticity problem in (\ref{sfp2}) is a solution of the inverse problem. Since we do not impose any conditions on the measurements $\vect u_0$, the solvability of (\ref{eq:if}) cannot be ensured. The methodology presented in this paper aims to find a traction and displacement field that solves a discrete version of the inverse problem, and if no solution exists, minimises the error $\bar a(\bu_1, \vect v)+\bar c(\bv)-\bar b(\vect t, \vect v)$ for arbitrary test functions $\vect v$. 

Furthermore, $\bu_0$ is in practice only retrieved on a set of $n_0$ discrete points $X=\{\bx_1, \ldots, \bx_{n_0}\}$ of $\Omega_0$. We denote by $O$ the operator that extracts the values of a  continuous field $\bu_0$ on the set $X$, that is, $O\bu_0=\{\bu_0(\vect x_1), \ldots, \bu_0(\vect x_{n_0})\}$. The inverse problem in (\ref{eq:if}) is then modified by defining the following functional,
\begin{eqnarray}\label{J0}
\tilde J_0(\bt, \bu_1):=||O\bu[\vect t, \bu_1]-O\bu_0||^2,
\end{eqnarray}
with $||\bullet||$ the standard Euclidean norm in $\Rset^{3\times n_0}$, and solving the following minimisation problem:
\begin{eqnarray}\label{mJ0}
(\bt^*, \bu_1^*)=\mbox{\textunderset{$\bt,\bu_1$}{$\argmin$}}\ \tilde J_0(\bt, \bu_1).
\end{eqnarray}

We note that this minimisation problem differs from other inverse problems which also consider partial knowledge of the field $\bu_0$ \cite{ambrosi06,salsa08}. The aim in these works is to minimise the following regularised functional :
\begin{eqnarray}\label{Je}
\tilde J_\epsilon(\bt):=||O\bu[\bt]-O\bu_0||^2 + \epsilon||\bt||^2_{2,\Gamma_n},\ \epsilon>0,
\end{eqnarray}
with $||\bullet||_{2,\Gamma_n}$ the $L^2$-norm in $\Gamma_n$. Our functional in (\ref{J0}) is instead non-regularised, that is, $\tilde J_0(\bt, \bu_1)$ in (\ref{J0}) does not include the term $\epsilon||\bt||^2_{2,\Gamma_n}$. This term is needed in order to ensure the coercivity of the penalty functional $\tilde J_\epsilon(\bt)$, and therefore guarantee the uniqueness of the optimum $\bt^*$ (see for instance \cite{salsa08}, Section 8.9, for a proof). If this term is not included, the minimisation problem may become ill-posed, and the solution of its discrete form may require the computation of a pseudo-inverse matrix, which may become computationally prohibitive. However, in the functional defined in (\ref{Je}), the value of the parameter $\epsilon$ needs to be chosen appropriately \cite{lions68}. The larger the value of $\epsilon$, the larger the error in the equilibrium equations in (\ref{wf}) becomes, while for very small values of $\epsilon$, the regularised problem may become ill-conditioned  \cite{hansen00}. In the next section, instead of considering the regularisation of the problem in  (\ref{J0})-(\ref{mJ0}), we will analyse the discrete form of the inverse problem and study the need for such regularisation.



\section{Finite Element discretisation}\label{s:disc}

\subsection{Discrete direct problem}\label{s:df}

Let us consider a finite element discretisation of the weak form in (\ref{wf}). We discretise the domain $\Omega\subset\Rset^3$ with a structured or unstructured mesh $M$ using $N_E$ non-overlapping conformal hexahedral elements $K_1, \ldots, K_{N_E}$ and $N$ nodes $\vect x_i\in\Omega, i=1,\ldots, N$. Elements $K_e, e=1,\ldots, N_E$ are such that \cite{beilina14, ciarlet88},
\begin{equation*}
\bar\Omega=K_1\cup K_2\ldots\cup K_E, \quad \mathrm{int}(K_e)\cap \mathrm{int}(K_{e'})=\emptyset, \forall e\neq e'.
\end{equation*}

 Let us define the polynomials $Q_r(K_e)$ on an element $K_e$ as
\begin{eqnarray*}
Q_r(K_e)=\{q: q(x,y,z)=\\
\quad\quad\quad\quad\quad\quad\sum_{0\le i,j,l\le r} c_{ijl} x^i y^jz^l, (x,y,z)\in K_e, c_{ijl}\in\Rset, \forall K_e\in M\}.
\end{eqnarray*}

In our numerical examples we will use the case $r=1$, which is tantamount  to using the following finite element spaces $U^h\subset U$ and $V^h\subset V$:
\begin{eqnarray*}
U^h=V^h=\{\vect v(\vect x)\in (H^1_0(\Omega))^3: \vect v\in C(\Omega), \vect v|_{K_e}\in Q_1(K_e),\ \forall K_e\in M\}.
\end{eqnarray*}

We are thus employing tri-linear hexahedral elements, although the results presented here are also valid for other element types and degrees.  After replacing in (\ref{wf}) the spaces $U$ and $V$ by $U^h$ and $V^h$, respectively, the discrete version of the weak form reads,
\begin{eqnarray}\label{wfh}
\mbox{Find}\ \bu^h\in U^h\ \mbox{s.t.}\ a(\bu^h,\bv^h)=b(\bv^h), &\ \forall \bv^h\in V^h.
\end{eqnarray}

The space of the traction field $T$ will be replaced by the set of piece-wise bi-linear tractions in $C^0(\Gamma_n)$:
\begin{eqnarray*}
T^h:=\{\vect t(\vect x)\in (L^2(\Gamma_n))^3: \vect t|_{\partial K_e}\in Q_1(\partial K_e), \forall \ \partial K_e\in \Gamma_n\}.
\end{eqnarray*}

The space $T^h$ is illustrated  in Figure \ref{f2}b, together with the nodally interpolated displacement field $\bu^h$. The use of the spaces defined above is equivalent to resorting to the following Lagrangian interpolation of the field $\vect u$ and traction field $\bt$,
\begin{eqnarray*}
\vect u\approx\vect u^h&=\sum_{\forall j, \vect x_j\not\in \Gamma_d} q_j(\vect x)\mbf u_j,\\
\bt\approx\bt^h& =\sum_{\forall j, \vect x_j\in\Gamma_n} \tilde q_j(\vect x)\bmt_j, 
\end{eqnarray*}
where the polynomials $q_j(\vect x)\in Q_1$ and $\tilde q_j(\vect x)$ form bases of the spaces $U^h$ and $T^h$ respectively, and  $\bsmu_j\in\Rset^3$  and $\bmt_j \in\Rset^3$ are  the displacement and traction vectors at node $j$. The solution of (\ref{wfh}) is then equivalent to solving the following system of equations \cite{hughes87}:
\begin{eqnarray}\label{eq:fp}
\bK\bsmu=\bA\bmt,
\end{eqnarray}
with $\bK$ the standard stiffness matrix and $\bA$ a matrix that projects the boundary loads on nodal contributions. Matrix $\bK$ is formed by block matrices $\bK_{ij}$ that couple nodes $i$ and $j$,
\begin{eqnarray*}
\mbf K_{ij}=\int_\Omega \left(\lambda \nabla q_i\nabla q_j^T+\mu((\nabla q_i^T\nabla q_j)\mbf I+\nabla q_j\nabla q_i^T)\right)\mathrm d\Omega,\  \vect x_i, \vect x_j\not\in \Gamma_d,\\
\end{eqnarray*}
while matrix $\bA$ adopts the following expression: 
\begin{eqnarray*}
\mbf A_{ij}=\mbf I\int_{\Gamma_n} \tr(q_i(\vect x))\tilde q_j(\vect x) \mathrm d \Gamma,\ \forall i,j, \vect x_i\not\in \Gamma_d, \vect x_j\in\Gamma_n,
\end{eqnarray*}
with $\tr(q_i(\vect x))$ the trace of function $q_i(\vect x)$ on the domain $\Gamma_n$. Vectors $\bsmu$ and $\bmt$ in (\ref{eq:fp}) gather the set of nodal displacements $\bsmu_j, \forall\vect x_j\not\in\Gamma_d$, and nodal values $\bmt_j$.

\begin{figure}[!htb]
\centerline{\subfigure[]{\includegraphics[width=0.5\textwidth]{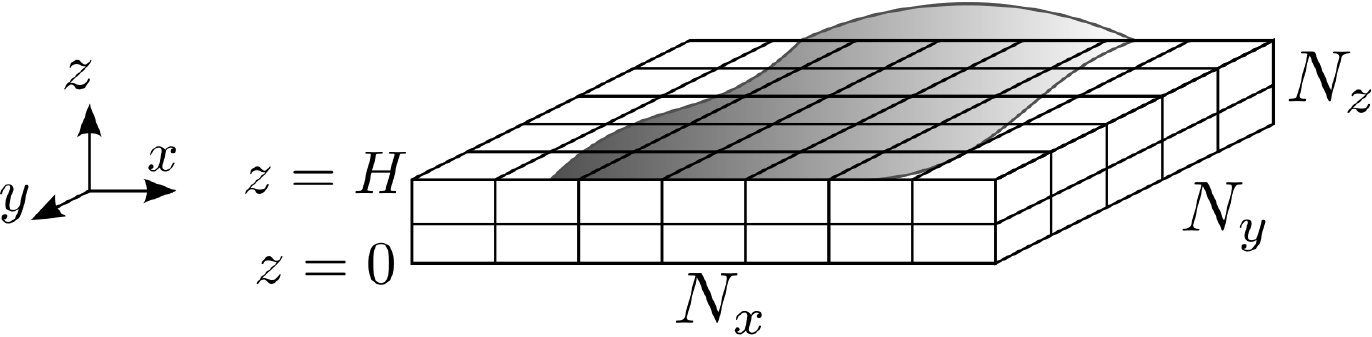}}\hspace{6ex}
\subfigure[]{\includegraphics[width=0.25\textwidth]{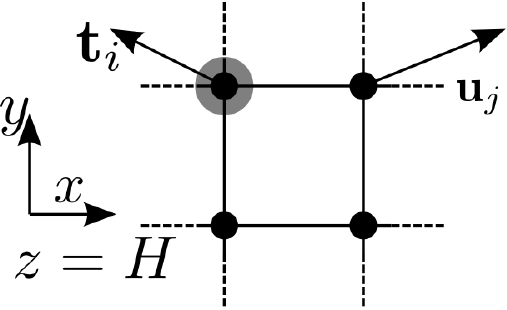}}}
\caption{a) Scheme of cell monolayer on gel substrate, discretised with a three-dimensional two-layered Cartesian mesh. b) Scheme of the discretisation on the top surface of the substrate, using nodal displacements and nodal tractions. Black circles and lines indicate respectively nodes and finite element mesh. The grey circle indicate the location where the traction field is defined.}
\label{f2}
\end{figure}

It will be convenient  to consider a modified matrix $\hat\bA$ and alternative loading consisting on a set of point loads $\hat\bmt_j, \forall \vect x_j\in\Gamma_n$:
\begin{eqnarray}\label{eq:hA}
\hat\bA_{ij}=\delta_{ij}\mathbf I,\ \forall i,j, \vect x_i\not\in\Gamma_d, \vect x_j\in\Gamma_n.\\
\hat\bmt_j=\int_{\Gamma_n} \tr(q_i(\vect x))\bt^h\mathrm d \Gamma,\ \forall j, \vect x_j\in\Gamma_n, \bt^h\in T^i.\label{eq:ht}
\end{eqnarray}
 
It can be verified that $\bA\bmt=\hat \bA\hat\bmt$, and therefore the product $\hat\bA\hat\bmt$ has  equivalent total nodal loads, but with simpler matrix components $\hat\bA_{ij}$. Furthermore, and in view of equation (\ref{eq:ht}), the relation between $\bmt$ and $\hat\bmt$ may be written as
\[
\hat\bmt=\bmM\bmt,
\]
with
\[
\bmM_{ij}=\mathbf I \int_{\Gamma_n} \tr(q_i(\vect x))\tilde q_j(\vect x)\mathrm d \Gamma, \forall i,j, \vect x_i, \vect x_j\in\Gamma_n.
\]

The symmetric matrix $\bmM_{ij}$ corresponds in fact to the mass matrix\red{, but with no density factor,} associated to the boundary $\Gamma_n$, and is thus invertible \cite{hughes87}. Therefore, the space $T^h$ can  be represented by either $\hat\bmt$ or $\bmt=\bmM^{-1}\hat\bmt$. In the former case, the system of equations of the discrete direct problem in (\ref{eq:fp}) takes the following form:
\begin{equation}\label{eq:fpb}
\bK\bsmu=\hat\bA\hat\bmt,
\end{equation}
with $\hat\bA=\bA\bmM^{-1}$ the matrix given in (\ref{eq:hA}). Since we assume that the  continuous problem in (\ref{wf}) has a unique solution, the discrete problem (\ref{wfh}) has also a unique solution \cite{ciarlet88}, that is, matrix $\bK$ is regular, and the system of equations in (\ref{eq:fp}) and in (\ref{eq:fpb}) accept the same unique solution $\bsmu$ for any loading $\bmt$ and $\hat\bmt=\bmM\bmt$. 

In agreement with the partitioning of domain $\Omega=\Omega_0\cup\Omega_1$ presented in Section \ref{s:cf}, with $\mathrm{int}(\Omega_0)\cap\mathrm{int}(\Omega_1)=\emptyset$, we will also consider the following decomposition of the discretised displacement field $U^h=U^h_0\oplus U^h_1$, where $U^h_0$ denotes the set of nodal measured displacements $\bsmu_0$, and $U^h_1$ is the set of nodal unknown displacements $\bsmu_1$, which cannot be experimentally measured. For instance, for the geometry depicted in Figure \ref{f2}a, $U^h_1$ may include the vertical displacements at the top surface, $\bsmu_z$, or in multilayered discretisations ($N_z>1$), the displacements at the intermediate layers of the gel $\bsmu_m$ at height $z$, $0<z<H$.

Let us introduce the notation $n_0=\dim(U^h_0)$, $n_1=\dim(U^h_1)$, $n=\dim(U^h)=n_0+n_1$,  and $m=\dim(T^h)$. According to the  partitioning of $U^h$, the system of equations corresponding to the direct problem reads:
\begin{equation}\label{eq:fp1}
\bK_0\bsmu_0 + \bK_1\bsmu_1=\bA\bmt,
\end{equation}
or equivalently,
\begin{equation}\label{eq:fp1b}
\bK_0\bsmu_0 + \bK_1\bsmu_1=\hat\bA\hat\bmt,
\end{equation}
with $\bK_0\in\Rset^{n\times n_0}$, $\bK_1\in\Rset^{n\times n_1}$, and $\bA,\hat\bA\in\Rset^{n\times m}$. 

We note that partial knowledge of $U$ has been also treated in \cite{ambrosi06,vitale12,vitale12b} by resorting to the adjoint problem  of the continuous problem \cite{lions68}, which is solved together with a regularised inverse problem. Here instead, we deal with this partially known displacement in the discretised problem.

In addition to the discrete equilibrium equations in (\ref{eq:fp}),  we could also impose  the  equilibrium conditions on the discrete traction field $\int_\Gamma \bmt^h d\Gamma = \mathbf 0$ and  $\int_\Gamma \bx\times\bmt^h d\Gamma = \mathbf 0$. These conditions are not considered here because in many cases, due to experimental limitations, the cell monolayer is not  isolated, and just a subset of the whole cellular system can be analysed. In this case,  the boundary $\Gamma_n$ includes external forces that other cells not in $\Omega$ exert.

\subsection{Discrete inverse problem}\label{s:di}

The discretised form of the functional $\tilde J_0(\bt, \bu_1)$ in (\ref{J0}) reads:
\begin{equation}\label{eq:Jt0h}
\tilde J_0^h(\bmt, \bsmu_1):=||\bK_0^{-1}\left(\bA\bmt-\bK_1\bsmu_1\right)-\bsmu_0||^2.
\end{equation}

The minimisation problem $\min_{\bmt, \bsmu_1} \tilde J_0^h(\bmt, \bsmu_1)$ would give rise to a system of equations that requires the computation of $\bK_ 0^{-1}$. For this reason, we will instead use the following functional,
\begin{eqnarray}\label{J0h}
J_0^h(\bmt, \bsmu_1):=||\bK_0\bsmu_0+\bK_1\bsmu_1-\bA\bmt||^2,
\end{eqnarray}
which may be interpreted as the functional in (\ref{eq:Jt0h}) but using a different metric of the vector space. 
The minimisation problem $\min_{\bmt,\bsmu_1} J_{0}^h(\bmt,\bsmu_1)$ gives now rise to the following normal equations:
\begin{eqnarray}\label{eq:ne2}
\left[\begin{array}{cc}
\bA & -\bK_1\end{array}\right]^T
\left[\begin{array}{cc}
\bA & -\bK_1\end{array}\right]
\left\{\begin{array}{c}
\bmt\\
\bsmu_1
\end{array}\right\}
=\left[\begin{array}{cc}
\bA & -\bK_1\end{array}\right]^T
\bK_0\bsmu_0.
\end{eqnarray} 

It will become convenient to rewrite this system as the solution of variables $\bmt$ and $\bsmu_1$ in a partitioned manner. By using the relation $\bA=\hat\bA\bmM$, and the fact that the mass matrix $\bmM$ is positive definite and thus invertible, the system of equations in (\ref{eq:ne2}) can be rewritten as,
\begin{eqnarray}
\bK_1^T\tmI\bK_1\bsmu_1 &=&-\bK_1^T\tmI\bK_0\bsmu_0, \label{eq:uzp}\\
\phantom{\bK_1^T\tmI\bK_1}\bmt&=&\phantom{-}\bmM^{-1}(\hat\bA^T\hat\bA)^\dagger\hat\bA^T(\bK_0\bsmu_0+\bK_1\bsmu_1), \label{eq:uzbp}
\end{eqnarray}
where $\tmI=\mathbf I-\hat\bA(\hat\bA^T\hat\bA)^\dagger\hat\bA^T$, and $(\hat\bA^T\hat\bA)^\dagger$ denotes the pseudo-inverse of $\hat\bA^T\hat\bA$,  which is equal to the inverse of $\hat\bA^T\hat\bA$ when this matrix is invertible. This new form allows us to compute $\bsmu_1$ from equation (\ref{eq:uzp}), and  then obtain the traction field $\bmt$ using equation (\ref{eq:uzbp}).

The form in (\ref{eq:uzp})-(\ref{eq:uzbp}) is clearly more convenient because requires to solve only the system in (\ref{eq:uzp}). It will be also shown in the numerical results that the condition number of matrix $\bK_1^T\tmI\bK_1$ is much lower than the matrix of the system in (\ref{eq:ne2}). The next proposition analyses the uniqueness of the solution in the normal equations (\ref{eq:ne2}) or (\ref{eq:uzp})-(\ref{eq:uzbp}), and also determines when the computation of the pseudo-inverse is necessary.

\begin{proposition}\label{p1}
\begin{itemize}
\item[\emph{i)}]The vectors of nodal tractions $\bmt\in \Rset^m$ and displacements $\bsmu_1\in\Rset^{n_1}$ that satisfy the system of equations in (\ref{eq:ne2}) are unique if and only if $m \le n_0$.
\item[\emph{ii)}] If $m=n_0$, the optimal solution $(\bmt^*,\bsmu^*)$ satisfies $J^h_{0}(\bmt^*, \bsmu^*_1)=0$.
\end{itemize}
\end{proposition}

\begin{proof}
\begin{itemize}
\item[\emph{i)}] \emph{Only if} implication. We will show that when $m>n_0$, the solution is not unique. We will distinguish two situations:

$m>n_1+n_0$: Matrix $\bA\in\Rset^{n\times m}$ is rectangular with $m>n$, and therefore $\dim(\ker(\bA))=\dim(\ker(\bA^T\bA))\ge m-n>0$. Consequently, either (\ref{eq:ne2}) has no solution, or if $\bmt^*$ is a solution of (\ref{eq:ne2}), any solution with the form $\bmt^*+\alpha\bmt_0$, $\alpha\neq 0$, with $\bmt_0\in\ker(\bA)$,  will also be a solution of (\ref{eq:ne2}). 

$n_0+n_1\ge m>n_0$: Matrix $\bK_1\in\Rset^{n\times n_1}$ stems from a FE discretisation of a linear elastic problem, with its columns associated to the displacements $\bsmu_1\in\Rset^{n_1}$. Since the non-discretised problem in (\ref{sf1})-(\ref{sf3}) is well-posed, matrix $\bK_1$ is full-rank and thus $\dim(\ker(\bK_1))=0$. Also, when $m\le n$, we have that $\hat\bA^T\hat\bA=(\hat\bA^T\hat\bA)^\dagger=\mathbf I\in\Rset^{m\times m}$, and therefore, $\tmI=\mathbf I-\hat\bA\hat\bA^T\in\Rset^{n\times n}$ is an identity matrix with $m$ diagonal components equal to zero. The product $\tmI\bK_1$ is then equal to matrix $\bK_1$, but with $m$ rows being equal to $0$. If $m>n_0$, then $\dim(\ker(\tmI\bK_1))\ge m-n_0>0$. This implies that $\tmI\bK_1$ is rank-deficient and that the system of equations in (\ref{eq:uzp}) has no solution or accepts more than one solution. In the latter case, by using equation (\ref{eq:uzbp}) for any of these multiple solutions, we obtain in turn multiple traction vectors $\bmt$.

\emph{If} implication. As before, $\dim(\ker(\bK_1))=0$, and matrix $\tmI\bK_1$ is equal to matrix $\bK_1$ with $m$ rows being replaced by $0$. The $m$ rows correspond to those degrees of freedom (dofs) where the tractions are applied. Since $m\le n_0$, $\bmt$ is applied onto nodes where $\bsmu$ is known, i.e. a subset of $\bsmu_0$. Hence, $\dim(\ker(\tmI\bK_1))=\dim(\ker(\bK_1))=0$. It follows that the  system of equations in (\ref{eq:uzp}) has full-rank, and the solution $\bsmu_1$ is unique. The optimal traction field $\bmt$ is also unique after inserting $\bsmu_1$ into (\ref{eq:uzbp}). 
\item[\emph{ii)}] When $m=n_0$, we have that $\hat\bA\hat\bA^T=\mathbf I\in\Rset^{n\times n}$. Then, by inserting the expression of $\bmt$ in (\ref{eq:uzbp}) into
\[
\bK_0\bsmu_0+\bK_1\bsmu_1-\bA\bmt
\]
and using the relation $\bA=\hat\bA\bmM$, it can be verified that the expression above vanishes, and therefore, the functional $J^h_{0}(\bmt, \bsmu_1)$ defined in (\ref{J0h}) also vanishes.
\end{itemize}
\end{proof}

In view of Proposition \ref{p1} and its proof, we can conclude that when $m\le n_0$, the normal equations in (\ref{eq:uzp})-(\ref{eq:uzbp}) take the following simpler form:
\begin{eqnarray}
\bK_1^T\tmI_0\bK_1\bsmu_1 &=&-\bK_1^T\tmI_0\bK_0\bsmu_0, \label{eq:uz}\\
\phantom{\bK_1^T\tmI_0\bK_1}\bmt&=&\phantom{-}\bmM^{-1}\hat\bA^T(\bK_0\bsmu_0+\bK_1\bsmu_1), \label{eq:uzb}
\end{eqnarray}
with $\tmI_0=\bmI-\hat\bA\hat\bA^T$. If instead $m>n_0$, the inverse problem may be solved by resorting to the pseudo-inverse $(\hat\bA^T\hat\bA)^\dagger$ and the normal equations in (\ref{eq:uzp})-(\ref{eq:uzbp}), or alternatively, by using a regularised functional such as
\begin{eqnarray}\label{Jhe}
\tilde J^h_\epsilon(\bmt, \bsmu_1):=J_0^h(\bmt, \bsmu_1) +  \epsilon||\bmt||^2_2,\ \epsilon>0.
\end{eqnarray}

We note that Proposition \ref{p1} includes also the trivial case $n_1=0$, which is obtained by removing the components of vector $\bsmu_1$ and matrix $\bK_1$ in equation (\ref{eq:ne2}). The traction is then obtained from (\ref{eq:uzb}) simply as $\bmt=\bmM^{-1}\hat\bA^T\bK_0\bsmu_0$. However, the situation $n_1=0$ has no practical interest, since in general it furnishes a too inaccurate  solution of the discrete inverse problem.
 
\subsection{Condition of discrete inverse problem}\label{s:cond}

Although Proposition 1 allows to determine the conditions that ensure a unique solution, nothing has been said about the conditioning of the system of equations. In the trivial case $n_1=0$, the system in (\ref{eq:ne2}) reduces to equation (\ref{eq:uzb}), so that the conditioning of the system of equations is equivalent to the one of a standard direct FE problem. When $n_1>0$ and $m\le n_0$, the optimal values of the inverse problem require the solution of the system in (\ref{eq:uz}), whose condition number, $\kappa_I=cond(\bK_1^T\tmI\bK_1)$, depends on the finite element interpolation (degree, number of elements or their aspect ratio) and the Lam\'e parameters $\lambda$ and $\mu$.

In the numerical results presented in Section \ref{s:num2}, we show numerically the dependence of $\kappa_I$ on some relevant numerical parameters. We point out here that, as expected, $\kappa_I$ depends on the differences in element sizes within a problem, but it is independent of homogeneous variations of the element size $h$, since all the components in matrix $\bK_1$ and the eigenvalues are equally affected by $h$.

For a real symmetric square matrix $\bK$, we have that $cond(\bK^T\bK)=cond(\bK)^2$. In our case though, $\bK_1$ is a rectangular matrix with $n$ rows and $n_1$ columns, while $\tmI\bK_1$ is equal to $\bK_1$ but with $m$  rows equal to $0$. The full matrix of the direct problem is $\bK=[\bK_0\ \bK_1]$, whose condition number is equal to $\kappa_D=cond(\bK)$. We show in our numerical results that $\kappa_D <\kappa_I << \kappa_D^2$, that is, the condition number of the inverse problem is larger than the condition number of the direct problem, but does not worsen significantly for the examples tested, with up to $32000$ elements.

The condition number $\kappa_D$, and therefore also $\kappa_I$, depend on the number of elements. More importanty, $k_I$ also depends on the factor $m/n$. Indeed, as $m/n$ decreases, with $m\le n_0$, the matrix of the inverse problem resembles a direct problem, but with a matrix $\bK_1^T\tmI\bK_1$ that approaches $\bK^T\bK$. Therefore, we have the following relation:
\begin{equation*}
\lim_{m/n\rightarrow 0} \kappa_I=\kappa_D^2,
\end{equation*}
and in all the cases tested, we have that $\kappa_I<\kappa_D^2$.

\red{
The condition number of the normal equations affect in turn the stability  of the solution $\bsmu_1$ and $\bmt$. From the normal equation in \eqref{eq:uz}, and for a given perturbation $\delta\bsmu_0$ on the measured displacements $\bsmu_0$, the perturbation on the retrieved displacements $\delta\bsmu_1$ may be bounded as,}
\red{\begin{align}\label{eq:pertU}
\frac{||\delta\bsmu_1||_2}{||\delta\bsmu_0||_2}\le \kappa_I ||\bK_0||_2\frac{||\bsmu_1||_2}{||\hat \bA^T\bK_0\bsmu_0||_2}.
\end{align}}

\red{This relation follows from the fact that for any system with the form $\bmB\bmx=\bmb$, it can be shown that $ ||\delta\bmx||_2/||\delta\bmb||_2\le \kappa_B ||\bmx||_2/\bmb||_2$, with $\kappa_B$ the condition number of matrix $\bmB$, and by also using the relations,
\[
||\hat \bA^T\bK_0\delta\bsmu_0||_2\le||\bK_0\delta\bsmu_0||_2\le ||\bK_0||_2 ||\delta\bsmu_0||_2.
\]}

\red{A similar bound can be deduced from equation \eqref{eq:uzb}, which reads,}
\red{\begin{align*}
\frac{||\delta\bmt||_2}{||\delta\bsmu_0||_2}\le \kappa_M ||\bK_0||_2\frac{||\bmt||_2}{||\hat \bA^T\left(\bI-\bK_1\left(\hat\bA^T\bA\right)^{-1}\bK_1^T\tilde\bI_0\right)\bK_0\bsmu_0||_2},
\end{align*}}
\red{with $\kappa_M$ the condition number of the mass matrix $\bmM$. When $m=n_0$, we have that $\bsmu_1$ is absent and thus the previous relation simplifies to,}
\red{\begin{align}\label{eq:pertT}
\frac{||\delta\bmt||_2}{||\delta\bsmu_0||_2}\le \kappa_M ||\bK_0||_2\frac{||\bmt||_2}{||\hat \bA^T\bK_0\bsmu_0||_2}.
\end{align}}
 
\red{The bounds in \eqref{eq:pertU} and \eqref{eq:pertT} show that the condition numbers $\kappa_I$ and $\kappa_M$ determine the stability of the solution. The former may have a more detrimental effect, since it may attain the value $\kappa_D^2$. However, as noted before, in the cases tested we have that $\kappa_I<<\kappa_D^2$, and therefore the solution remains stable with respect to the noise in the measured displacements $\delta\bsmu_0$, as it also numerically verified in Example \ref{s:exp}. 
}

\subsection{Discussion and Implementation}

The results in Proposition \ref{p1} neither depend on the order of interpolation nor on its continuity, but just on the relative dimensions between the displacement and traction dofs. Therefore, these results are equally valid for other element types, as far as both the displacements and tractions are nodally interpolated. In fact, we note that the result in Proposition \ref{p1} does not carry over to the  situation when the traction field is considered as a piece-wise constant traction field. Although we do not study this case here, we just mention that for elementally interpolated tractions, the equivalent normal equations would yield a non unique solution if $m>n$, with $m$ the number of elemental dofs.

Proposition \ref{p1} allows us to state that the necessary condition $m\le n$ is not sufficient for obtaining a unique solution, and conclude that,

\begin{itemize}
\item Rank-deficient problems may be rendered full-rank by adding new measured displacements observations, or removing some of the nodal tractions, that is, by increasing $n_0$ or decreasing $m$. 
\item The problem that searches optimal nodal traction field  $\bmt$ on the same nodes where the displacement has been measured has a unique solution, since in this case $m=n_0$.
\end{itemize}

We note that from the definitions of the functionals in (\ref{J0h}), and assuming that the space $U^h=U_0^h\oplus U_1^h$ is constant, but with the partitions $U_0^h$ and $U_1^h$ changing dimensions, that is, keeping the mesh and $n$ constant, but changing $n_1$ and $n_0$, the following inequalities hold:
\begin{eqnarray}\label{eq:ineq}
\!\!\!\!
0\le\min_{\bmt,\bar\bsmu_1} J_{0}^h(\bmt,\bar\bsmu_1)\le\min_{\bmt,\bsmu_1} J_{0}^h(\bmt,\bsmu_1)\le \min_{\bmt, n_1=0} J_0^h(\bmt),\quad U^h_1\subseteq \bar U_1^h.
\end{eqnarray}

These relations open the possibility to design adaptive strategies, with the aim of
\begin{itemize}
\item  reducing the error of the least-squares problem, that is, the measure $J^h_0$ of the solution in the optimal inverse problem while keeping the mesh fixed.
\item reducing the error between the discrete and the analytical solution, that is, the norm of the difference between the discrete solution and the analytical solution, $||\bu^h-\bu||=(\bu^h-\bu,\bu^h-\bu)^{1/2}$.
\end{itemize}

\vspace{2ex}
\fbox{
\begin{minipage}[!htb]{0.9\textwidth}
\textbf{Box 1}. Solution algorithm for FEM based TFM.
\begin{itemize}
 \item \textbf{Step 1}. Build matrices $\bK_0, \bK_1$. Compute scalars $m$, $n_0$ and $n_1$.
\item \textbf{Step 2}. Compute matrices $\hat\bA$ and $\bmM$.
 \item \textbf{Step 3.}
 \begin{itemize}
  \item If $m > n_0$:
\begin{itemize}
\item[]  \textbf{Step 3.1.} Compute pseudo-inverse $(\hat\bA^t\hat \bA)^\dagger$ and solve (\ref{eq:uzp})-(\ref{eq:uzbp}).
\end{itemize}
  \item Else:
\begin{itemize}
\item[]   \textbf{Step 3.2.} Solve system in (\ref{eq:uzb})-(\ref{eq:uz}).
\end{itemize}
 \end{itemize} 
 \end{itemize} 
\end{minipage}}

\vspace{2ex}

In view of Proposition \ref{p1}, we can reduce the value of $J^h_0$ by increasing the ratio $m/n_0\le 1$. However, even in the case $m=n_0$, when the discrete equilibrium equations are exactly satisfied, the discrete solution $\bu^h$ may be too inaccurate with respect to the analytical solution $\bu$. For this reason, adaptive strategies for reducing the error $||\bu^h-\bu||$ should be envisaged. We will not apply these techniques here, but we point out that such strategies for inverse problems can be found for instance in \red{\cite{asadzadeh10,li11b,xu15}}, while other \emph{a posteriori} strategies for elasticity problems \cite{ainsworth00} could be used once the discrete solution $\bu^h$ is computed. 

In the numerical results given in the next section, we have applied the solution algorithm given in Box 1, with the spaces $U^h$ and $V^h$ specified in Section \ref{s:df}. The computation of pseudo-inverse matrix becomes necessary in Step 3.1. In this case, the regularisation of the inverse problem may become computationally more efficient than computing the pseudo-inverse. In these situations though, and according to the results deduced, it is advised to change the interpolation of the traction field or the displacement field in order to avoid regularising the inverse problem or computing the pseudo-inverse. In our numerical examples, the latter has been computed by retrieving the singular value decomposition of the system matrix (command {\tt svd} in Matlab).

\section{Numerical results}\label{s:res}

In all the examples tested here we have used a material with Young modulus $E=3000$ and Poisson ratio $\nu=0.3$ (Lam\'e constants $\lambda=1730.8$ and $\mu=1153.8$). The solution of the inverse problem has been implemented in Matlab 2013a. 
 
\subsection{Toy problem}

We have verified the previous results with a test problem on the domain $\Omega=\{(x,y,z) |0\le x,y\le 3, 0\le z\le 1\}$, with a $3\times 3$ mesh on the $x-y$ plane, and using one and two divisions along the height of the gel ($N_z=1,2$, see Figure \ref{f3}). This problem is too small to attract any  practical interest, but it is used here in order to verify the results in Proposition \ref{p1}.

Table \ref{t} summarises the 12 situations analysed. For each row, Table \ref{t} gives the dimension of the measured displacements $\bsmu_0\in\Rset^{n_0}$, and the dimension of the displacement computed through the inverse problem, $\bsmu_1\in\Rset^{n_1}$. The different values have been obtained by using different number of layers ($N_z=1,2$), prescribing some of the displacements, or additionally prescribing the vertical tractions $T_z$ on the top layer. The displacements were generated randomly, but are the same for all the cases considered. 

\begin{figure}[!htb]
\centerline{
\includegraphics[width=0.4\textwidth]{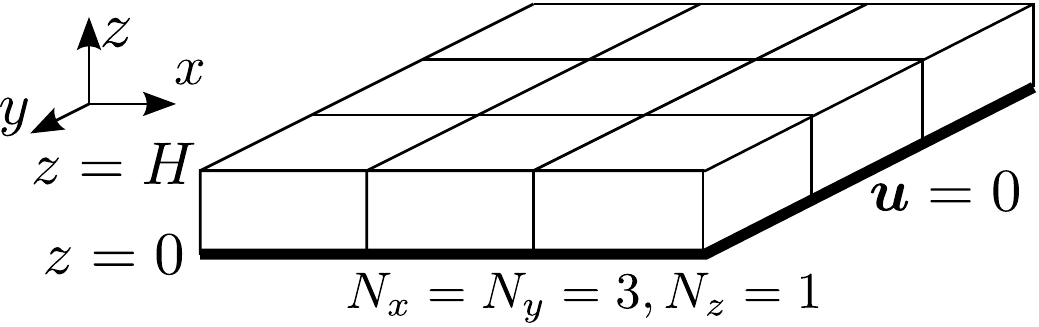}\hspace{6ex}
\includegraphics[width=0.4\textwidth]{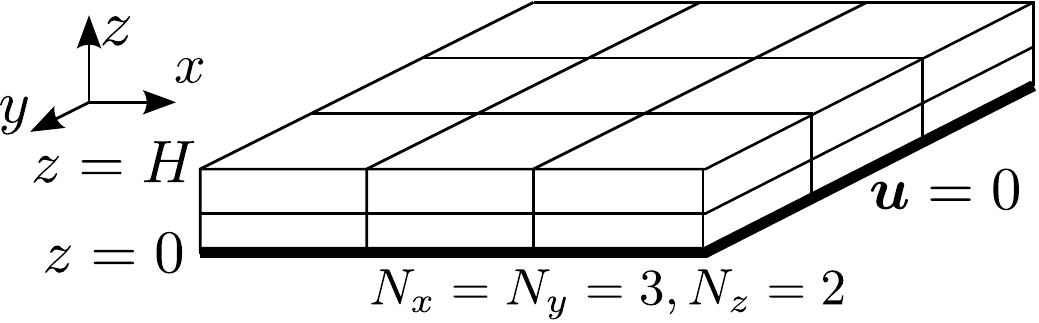}}
\caption{Toy problem on a $3\times 3$ grid with one (left) and two layers (right).}
\label{f3}
\end{figure}

The values in Table \ref{t} of the functional $J^h_0$ and the condition number $\kappa_I$  of the system being solved comply with the results in Proposition \ref{p1}. In the cases where the non-regularised inverse problem yields a singular matrix (indicated with $\kappa_I=\infty$), the value of $J^h_0$ has been computed from the solution of the pseudo-inverse (Step 3.1 in Box 1).  When $n_1=0$, we give the value  $\kappa_I=1$, since just the nodal values $\hat\bmt=\bmM\bmt$ are computed, and thus no system of equations is actually solved. When $n_1>0$, the values of $\kappa_I$ reported in Table \ref{t} correspond to the partitioned form (\ref{eq:uzb})-(\ref{eq:uz}). We note that the condition number of the equivalent  non-partitionned system with $n_0 + n_1$ unknowns,
\begin{equation*}
\left[\begin{array}{cc}
\hat\bA^T\hat\bA & -\hat\bA^T\bK_1\\
-\bK_1^T\hat\bA & \bK_1^T\bK_1
\end{array}
\right]
\left\{\begin{array}{c}
\bmM\bmt\\
\bsmu_1
\end{array}
\right\}
=\left\{\begin{array}{c}
\hat\bA^T\bK_ 0\bsmu_0\\
-\bK_1^T\bK_0\bsmu_0
\end{array}
\right\}.
\end{equation*}
oscillates between $2E7$ (case $b$) and $7E10$ (case $f$). These values are significantly higher than the condition number reported in Table \ref{t}, which highlights the advantage of solving the partitioned system of equations.

We remark that in all the cases where the inverse problem has full rank, and when the displacements $\bsmu_0$ are obtained from a direct FE problem, the tractions that produced them are fully recovered, that is, $J^h_0=0$. If the displacements are instead randomly generated, as it is the case in the results in Table \ref{t}, the optimal values of the functional $J^h_0$ are those indicated in the  table  (using always the same random displacements). 

When $n_1$ increases, and for constant spaces $U^h_0\oplus U_1^h$, as it occurs in cases $c-f$ and $i-l$ (due to constant boundary conditions and number of layers), $J^h_0$ decreases, in agreement with the inequalities in (\ref{eq:ineq}). In addition, when $n_1$ increases, with $n_0$ constant (see cases $a$ and $e$), the value of $J^h_0$ diminishes. This trend shows that the error in the mechanical equilibrium is reduced as $n_1$ increases, even if no traction field satisfying the discrete inverse problem exists. The evolution of this error is analysed further in the next section.

\begin{table}[!htb]
\centering
\begin{small}
\begin{tabular}{c c c c| c| c|c|c|c|}
\cline{3-7}
&&\multicolumn{2}{|c|}{\rule[-0.5ex]{0cm}{0.5cm}Tractions}
&\multicolumn{3}{|c|}{\rule[-0.5ex]{0cm}{0.5cm}Displacements}   \\
\hline
\mc{Case} & \mc{$N_z$} & \mc{$T_z$}&$m$&\rule[-0.5ex]{0cm}{0.5cm} $n_0$ & $n_1$ &$U_1$ & $\kappa_I{}^{*}$ & $J^h$ \\
\hline
\mc{$a$}&\mc{1\phantom{${}^*$}}&\mc{0} &32&48& 0& $\emptyset$     & $1$ & 0.918\\
\mc{$b$}&\mc{1\phantom{${}^*$}}&\mc{0} &32&32&16& $\bsmu_z$       & $50$ & 0 \\
\hline
\mc{$c$}&\mc{2\phantom{${}^*$}}&\mc{0} &32&96& 0& $\emptyset$     & $1$ & 5.583\\
\mc{$d$}&\mc{2\phantom{${}^*$}}&\mc{0} &32&80&16& $\bsmu_z$       & $260$ & 4.099\\
\mc{$e$}&\mc{2\phantom{${}^*$}}&\mc{0} &32&48&48& $\bsmu_m$       & $620$ & 0.789\\
\mc{$f$}&\mc{2\phantom{${}^*$}}&\mc{0} &32&32&64& $\bsmu_{z,m}$   & $3E3$ & 0 \\
\hline
\hline
\mc{$g$}&\mc{1\phantom{${}^*$}}&\mc{unk} &48&48&0& $\emptyset$    & $1$ & 0\\
\mc{$h$}&\mc{1\phantom{${}^*$}}&\mc{unk} &48&32&16& $\bsmu_z$     & $\infty$ & 0\\
\hline
\mc{$i$}&\mc{2\phantom{${}^*$}}&\mc{unk} &48&96& 0& $\emptyset$   & $1$ & 5.239\\
\mc{$j$}&\mc{2\phantom{${}^*$}}&\mc{unk} &48&80&16& $\bsmu_z$     & $260$ & 3.603\\
\mc{$k$}&\mc{2\phantom{${}^*$}}&\mc{unk} &48&48&48& $\bsmu_m$     & $620$ & 0\\
\mc{$l$}&\mc{2\phantom{${}^*$}}&\mc{unk} &48&32&64& $\bsmu_{z,m}$ &  $\infty$ & 0\\
\hline
\end{tabular}

\end{small}
\vspace{1ex}
\caption{Results for the toy problem: substrate with $3\times 3$ divisions on the $x-y$ plane, and $N_z=1$ or $2$ divisions along $z$. $J^h=$ value of the optimal solution of the functional in equation (\ref{J0h}). In cases \emph{a-f}, the condition $T_z=0$ is assumed, while in cases \emph{g-l} component $T_z$ is unknown and is found using inverse analysis.
$^{*}$ $\kappa_I$ has been computed using the Matlab  function {\tt cond}. The case $\kappa=\infty$ means $\kappa>1E24$, in which case the pseudo-inverse was computed (Step 3.1 in Box 1).
}
\label{t}
\end{table}

\subsection{Analysis of condition number and error}\label{s:num2}

We will here evaluate the conditioning of the matrix in the inverse problem $\bK_1^T\tmI\bK_1$ using the same geometry given in Figure \ref{f3} but with different number of elements and boundary conditions. In order to not taking into account the dependence on the element aspect ratio, which would affect the condition number of the associated direct problem $\kappa_D=cond([\bK_0\ \bK_1])$ and thus also affect $\kappa_I=cond(\bK_1^T\tmI\bK_1)$, we have used solely cuboid elements, and adapted the height of the domain $H$ accordingly. 

As mentioned in Section \ref{s:cond}, the condition numbers $\kappa_I$ and $\kappa_D$ are independent of $h$, but they do depend on the number of elements $N_E$ and ratio $m/n$. Figure \ref{f4}a shows the evolution of $\kappa_I$ and $\kappa_D$ for different values of $N_E$, while keeping the ration $m/n$ constant. This is achieved by increasing $N_x$ and $N_y$, but keeping $N_z=constant$. It can be observed that $\kappa_I$ is slightly affected by $N_E$, overall for lower values of $m/n$.

\begin{figure}[!htb]
\centerline{
\subfigure[]{\includegraphics[width=0.55\textwidth]{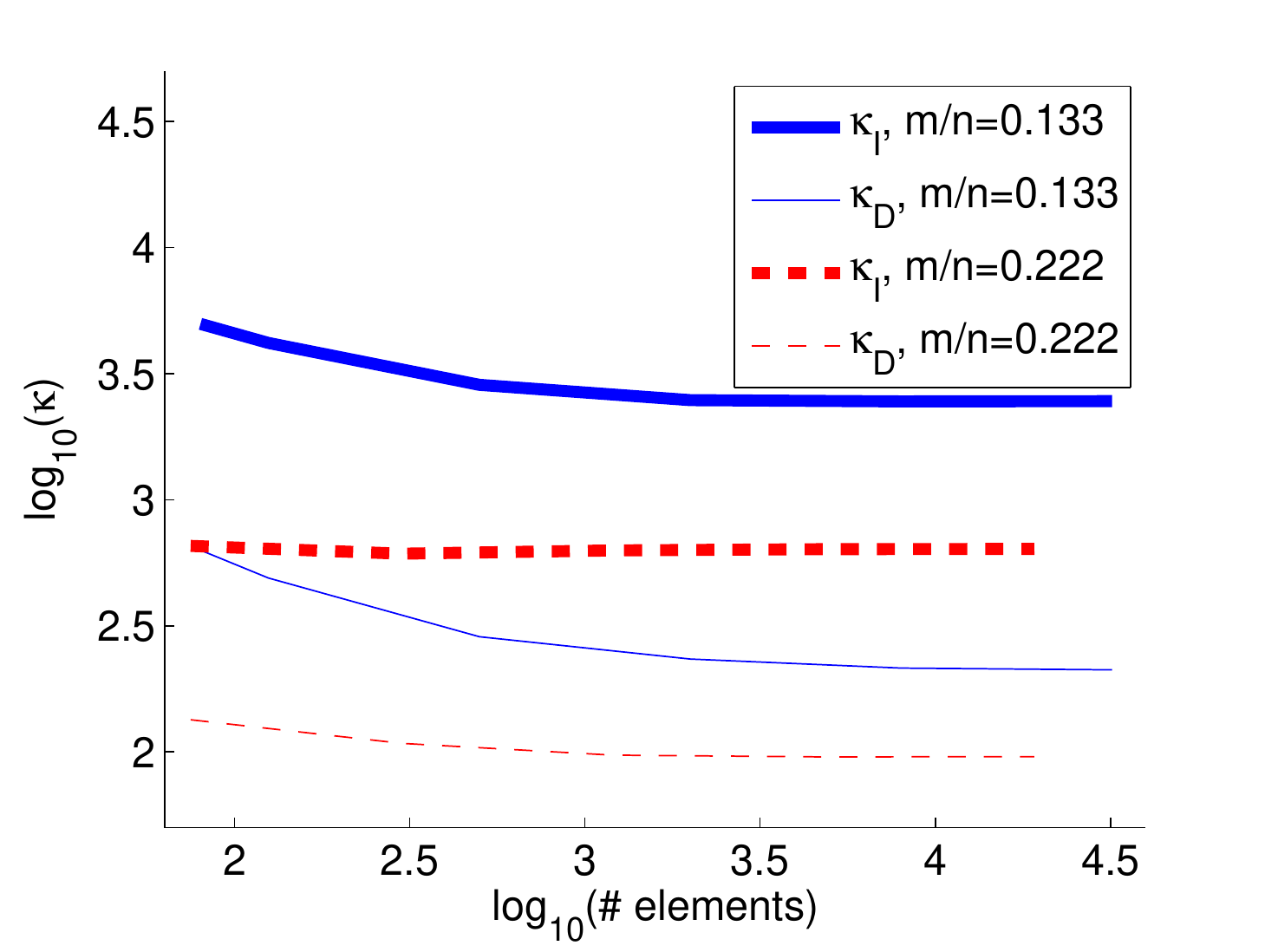}}\hspace{-5ex}
\subfigure[]{\includegraphics[width=0.55\textwidth]{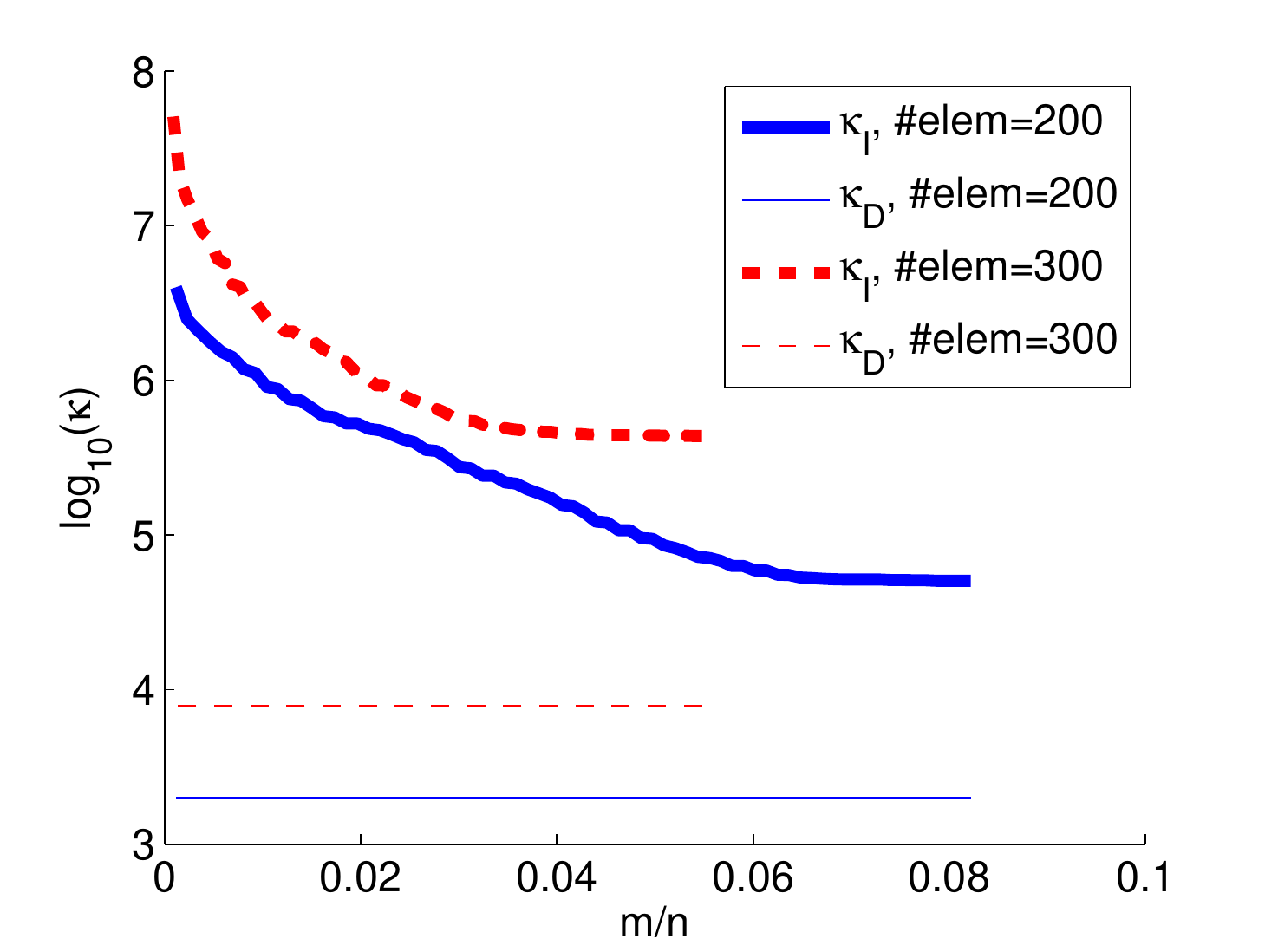}}}
\caption{Evolution of the condition numbers of the inverse and associated direct problem, $\kappa_I$ and $\kappa_D$ respectively, as a function of the number of elements and ratio $m/n$.}
\label{f4}
\end{figure}

Figure \ref{f4}b shows the evolution of the condition numbers for different values of $m/n \approx 1/N_z$, while keeping a constant number of elements $N_E=N_x* N_y *N_z$. We have analysed two sets of problems, one with $N_E=200$ and $N_E=300$ elements. Figure \ref{f4}b shows that $\kappa_I$ is indeed always larger than $\kappa_D$, and that as $m/n$ diminishes, $\kappa_I$ increases towards $\kappa_D^2$. However, the plot also shows that  this upper bound is approached only for very low values of $m/n$ (very heigh and narrow geometries when using cube-like elements). In more general flat-like geometries, we have that if $m/n>0.005$, then $\kappa_I/\kappa_D^2<0.5$, or that if $m/n>0.05$, then $\log_{10}(\kappa_I/\kappa_D)<1.5$, for the two sets of problems analysed.

\begin{figure}[!htb]
\centerline{\includegraphics[width=0.5\textwidth]{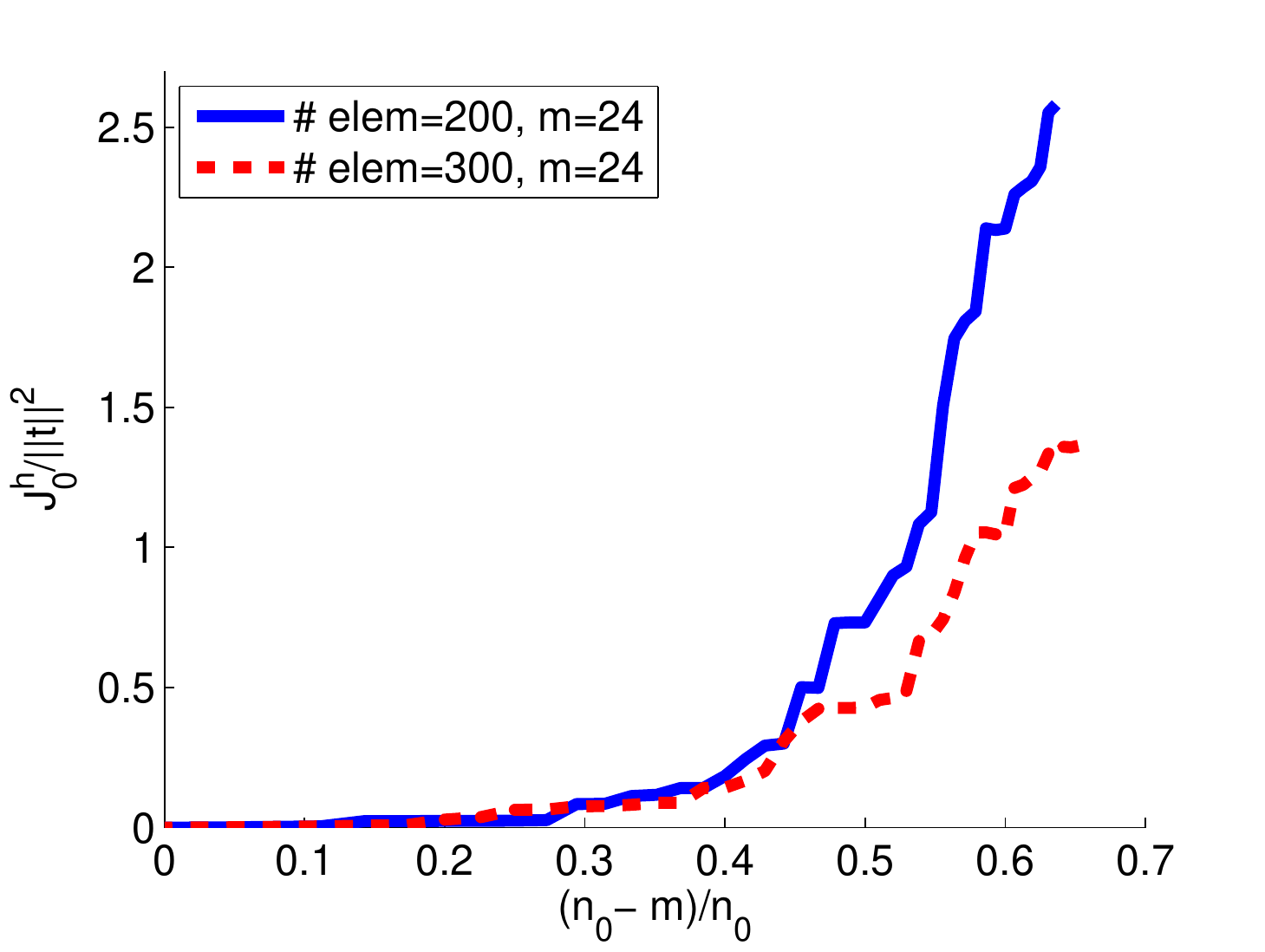}}
\caption{Evolution of the functional $J^h_0$ for different ratios of $(n_0-m)/n_0$}
\label{f5}
\vspace{-2ex}
\end{figure}

We have also measured the error of the discrete inverse problem by inspecting the evolution of the non-dimensional ratio $J_0^h/||\bmt||^2$ with respect to the ratio $(n_0-m)/n_0$, which for problems with a unique solution takes a value between $0$ and $1$. The converge rate is faster than linear, but slightly lower than quadratic. It can be observed in Figure \ref{f5} that $J^h_0$ may become larger than $0.5||\bmt||^2$ whenever $m<0.6n_0$, in which case the traction field becomes too poor with respect to the measured displacement field $\bsmu_0$.

\subsection{Experimental data}\label{s:exp}

We have also tested the algorithm with some real data of Madin Darby Canine Kidney (MDCK) II cells on a gel substrate with dimensions $(x,y)\in[0,55]\times [0,55]$ during wound healing. The displacements have been stored on a $56\times 56$ grid 72 minutes after wounding the tissue.  Figure \ref{f6}a shows the horizontal components of the displacements. Since these have been measured on the $56\times 56$ grid, a mesh with linear elements has been adapted to these locations for simplicity. We stress that if required, an irregular mesh or elements with higher degree could have been equally employed, without altering the methodology.

\begin{figure}[!htb]
\centerline{
\subfigure[]{\includegraphics[width=0.3\textwidth]{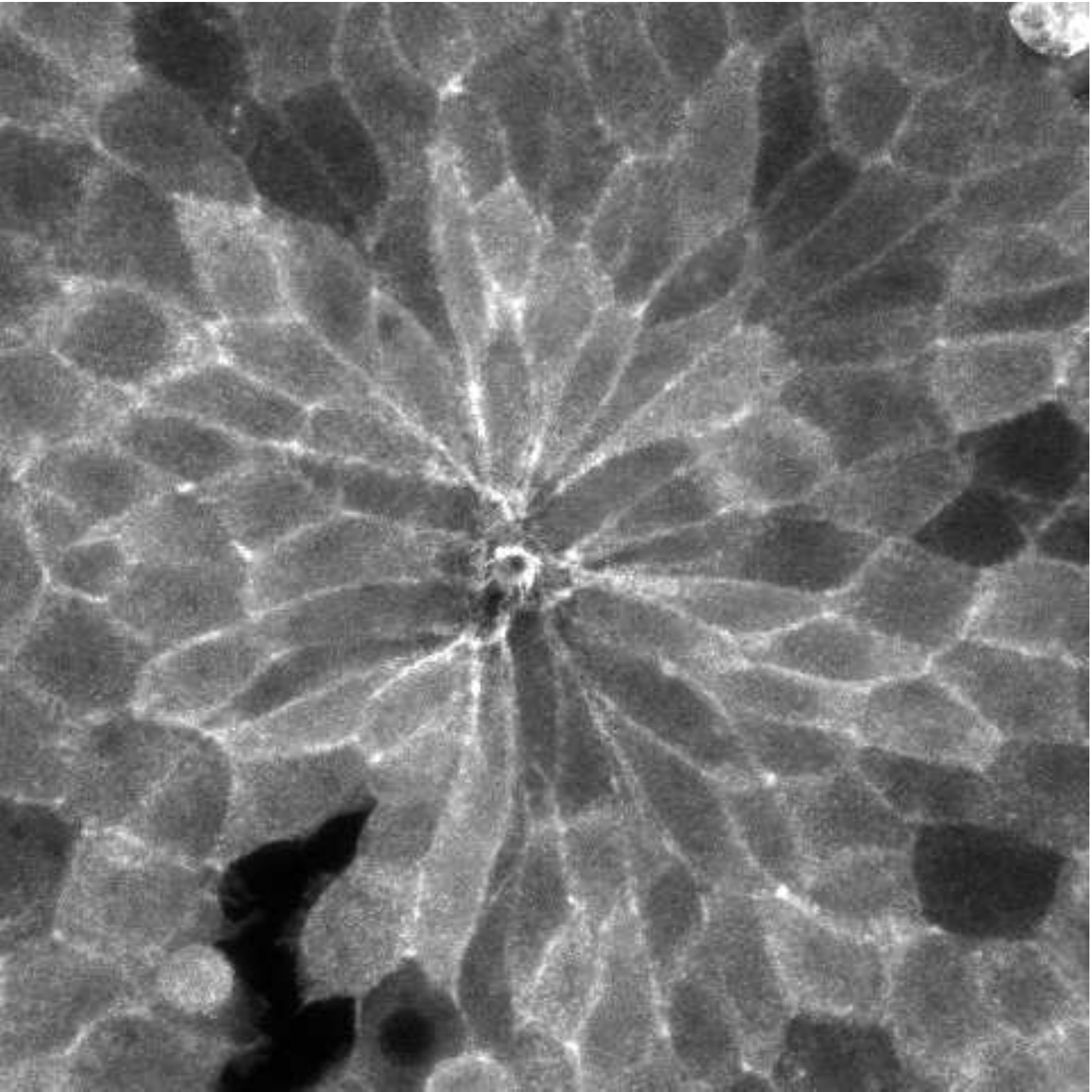}}
\subfigure[]{\includegraphics[width=0.45\textwidth]{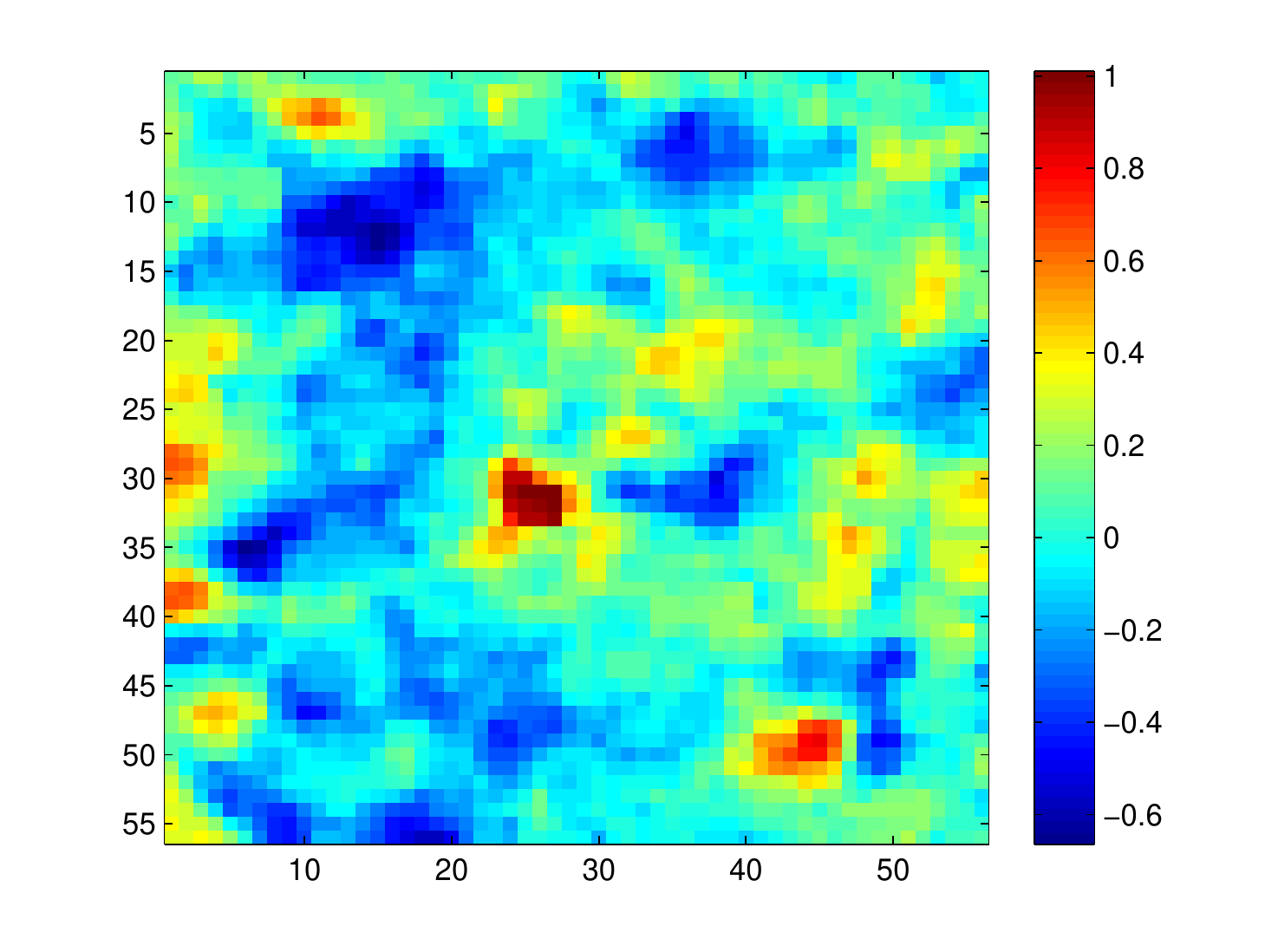}}}
\centerline{
\subfigure[]{\includegraphics[width=0.45\textwidth]{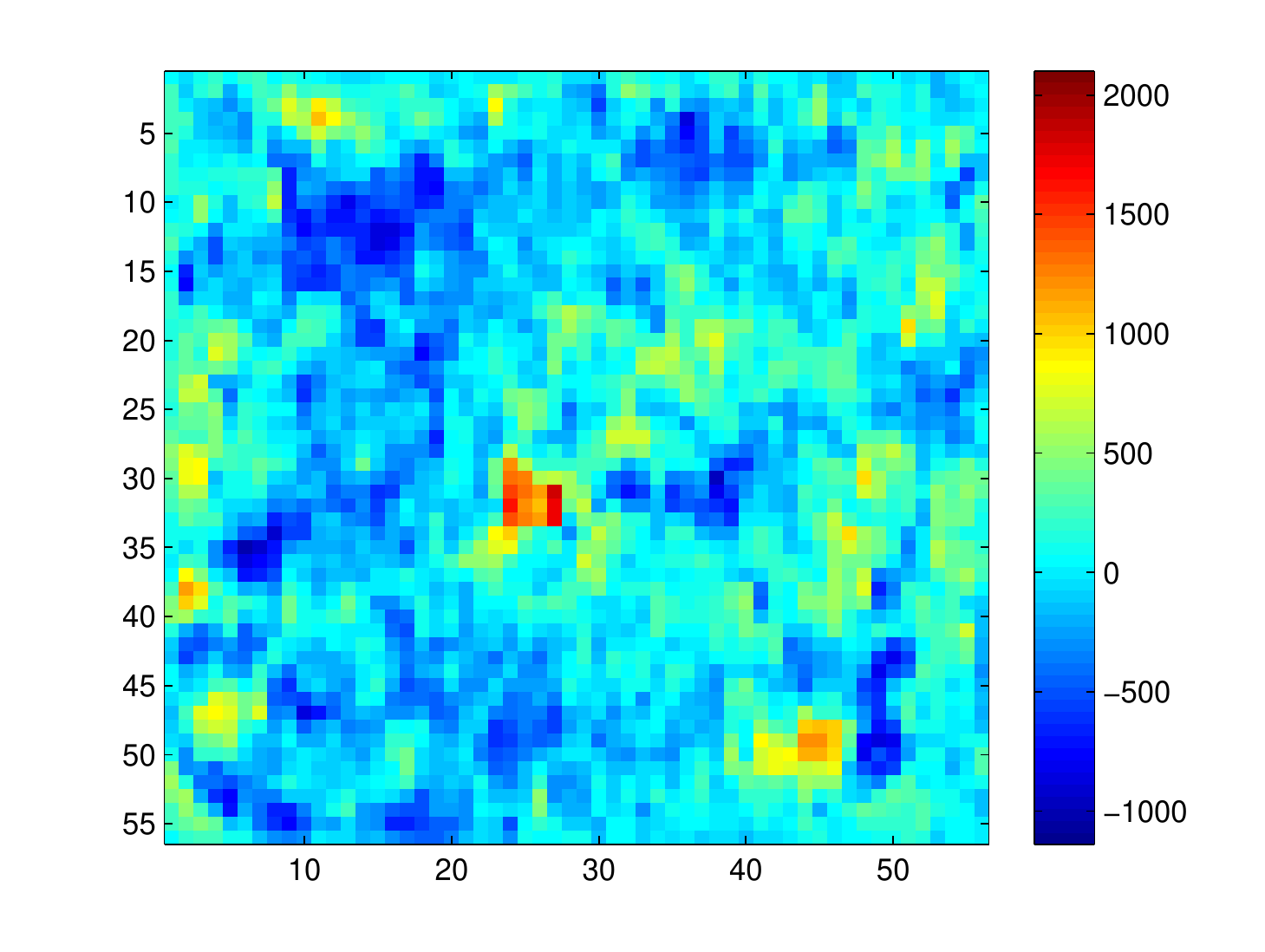}}\hspace{-0ex}
\subfigure[]{\includegraphics[width=0.45\textwidth]{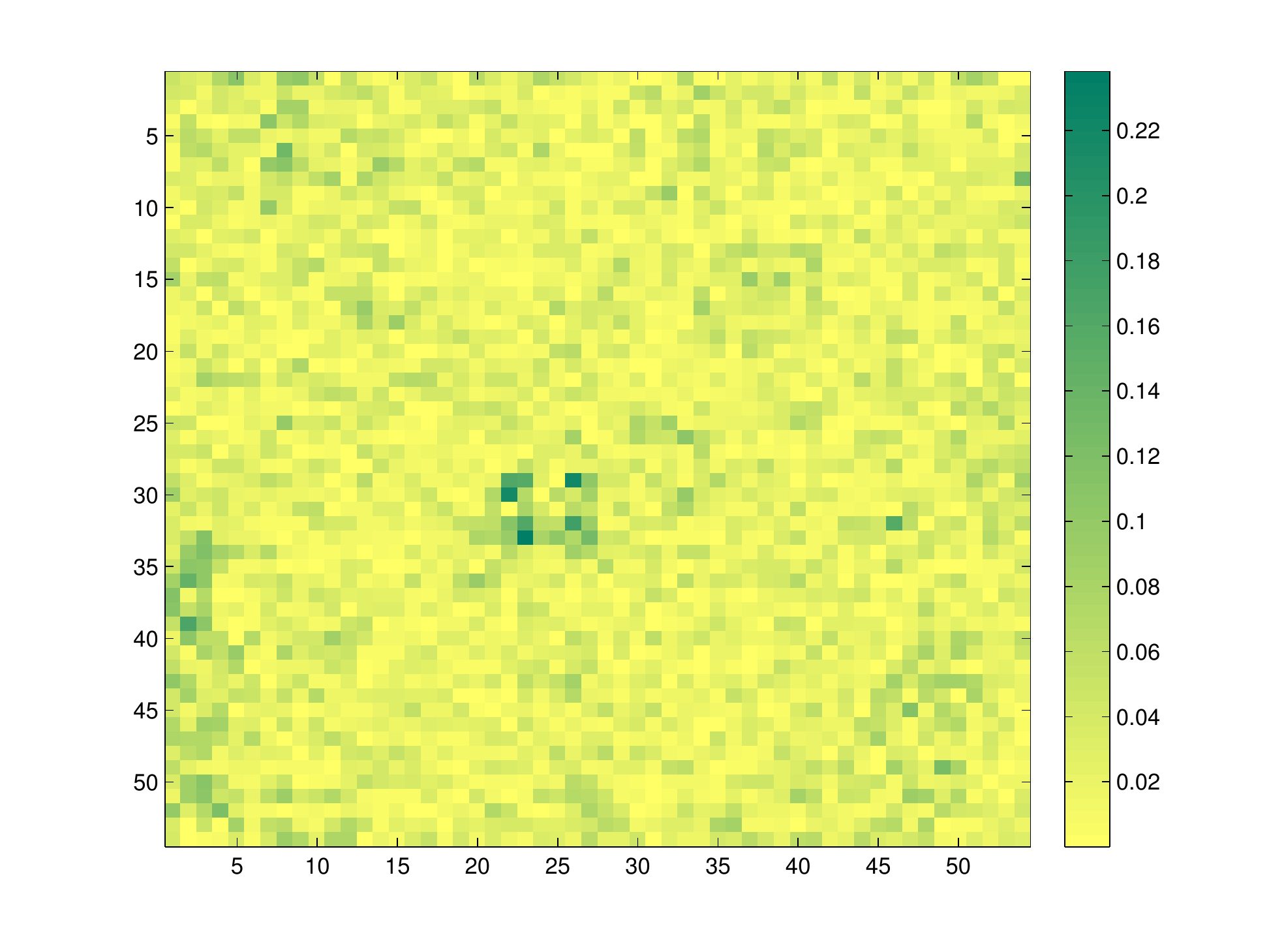}}}
\vspace{-2ex}
\caption{(a) Madin Darby canine kidney (MDCK) II cells 72 minutes after wounding (courtesy of Xavier Trepat \cite{brugues14}). (b) Horizontal displacement field $\bsmu_x$. (c) Horizontal traction field $\bmt_x$ using inverse finite element techniques with $N_z=1$ layer and a nodal traction field  $\bt^h\in T^i$. (d) Relative difference  between Boussinesq solution \cite{trepat09} and FE solution in (c) computed as $\delta_t=|\bmt_{x,Bous}-\bmt_{x,FE}|/\max\bmt_{x,Bous}$. The plot does not include the outer layer of elements. If these are included, the maximum value of $\delta_t$ increases from $0.24$ up to $0.49$.}
\label{f6}
\vspace{-2ex}
\end{figure}

Figure \ref{f6}c shows the resulting traction field resorting to the inverse FE analysis, $\bmt_{FE}$, computed with one layer and $m=n_0$, so that the equilibrium equations were exactly satisfied. Figure \ref{f6}d compares this solution and the Boussinesq solution $\bmt_{Bous}$ of an homogeneous elastic infinite half-plane \cite{landau67,trepat09} by showing the relative difference between the two, computed as  $\delta_t=|\bmt_{x,Bous}-\bmt_{x,FE}|/\max \bmt_{x,Bous}$. We note that the average of this difference is equal to $\bar\delta_t=0.035$, with a maximum value $\delta_{t,max}=0.49$.

The error between the two techniques is mainly due to the different interpolation in the displacements, and the different assumptions on the geometry (semi-infinite versus finite domain) and lateral boundary conditions (contact stresses due to the presence of material in Boussinesq versus zero tractions at the boundary in FE solution). The boundary effects may be reduced if for instance a one element band is excluded, which is where the errors are more pronounced. In this case, the averaged and maximum error are respectively reduced to $\bar\delta_t=0.027$ and $\delta_{t,max}=0.24$.

\red{We have also tested the bounds in \eqref{eq:pertU} and \eqref{eq:pertT} for the present case, by applying the  noise $\delta\bsmu_0$ on the displacements $\bsmu_0$, with $||\delta\bsmu_0||/||\bsmu_0||\approx 2E-8$, and  $\delta\bsmu_0$ a normal distribution with the same mean value than $\bsmu_0$. The resulting values on each side of the bounds are reported in Table \ref{t:pert}. It can be verified that the bounds are satisfied, and that in all cases we have that}

\red{
\[
\frac{||\delta\bmt||}{||\bmt||}\approx \frac{||\delta\bsmu_0||}{||\bsmu_0||}  \quad
;\quad \frac{||\delta\bsmu_1||}{||\bsmu_1||}\approx \frac{||\delta\bsmu_0||}{||\bsmu_0||},
\]}
\red{with $\bsmu_1$ and $\bmt$ the solutions for the unperturbed measure $\bsmu_0$. The computations of tractions and displacements remains thus stable with respect to the applied perturbations.}

\red{ We have also tested the effect of the perturbation for different magnitudes of $||\delta\bsmu_0||/||\bsmu_0||$ and also for different mesh sizes, using meshes with $6050$, $48400$, $163350$ and $387200$ elements, by using uniform element subdivisions. From the values in Table \ref{t:pertM} it can be verified that the relative size of the corresponding perturbed solutions, measured by $||\delta\bsmu_1||/||\bsmu_1||$ and $||\delta\bmt||/||\bmt||$, are not  affected by the element and mesh size.}

\begin{table}[!htb]
\centering
\red{
\begin{tabular}{|c ||c |c ||c| c||c| c| c| }
\hline
$N_z$\phantom{$\Bigg\vert$} & 
$\frac{||\delta\bsmu_1||_2}{\ ||\delta\bsmu_0||_2}$ & {\small Rhs Eq.\eqref{eq:pertU}} &
$\frac{||\delta\bmt||_2}{||\delta\bsmu_0||_2}$ 
& {\small Rhs Eq.\eqref{eq:pertT}}
&  $\frac{||\delta\bsmu_0||}{||\bsmu_0||}$ & $\frac{||\delta\bsmu_1||}{||\bsmu_1||}$ 
& $\frac{||\delta\bmt||}{||\bmt||}$ \\
\hline
\hline
1\phantom{$\big\vert$}  & - & - & 1.15E3 & 1.95E5  &  1.94E-8 & - & 1.05E-8 \\
2\phantom{$\big\vert$} & 5.44E-1 & 1.85E3  & 9.75E2 & 3.36E5 &1.94E-8 & 2.25E-8 &1.12E-8 \\
3\phantom{$\big\vert$} & 8.07E-1 & 2.34E4  & 1.90E3 & 4.91E5 &1.94E-8 & 2.22E-8 &2.27E-8 \\
\hline
\end{tabular}}
\vspace{1ex}
\caption{\red{Verification of bounds in equations \eqref{eq:pertU} and \eqref{eq:pertT} for stability analysis of the experimental tests using $N_x=N_y=55$. }}
\label{t:pert}
\end{table}

\begin{table}[!htb]
\centering
\red{
\begin{tabular}{|c|c|c|| c| c| }
\hline
$N_z$\phantom{$\Bigg\vert$}  & $N_x=N_y$
&  $\frac{||\delta\bsmu_0||}{||\bsmu_0||}$ 
& $\frac{||\delta\bsmu_1||}{||\bsmu_1||}$ 
& $\frac{||\delta\bmt||}{||\bmt||}$ \\
\hline
\hline
2\phantom{$\big\vert$} &55& 1.94E-8 & 2.25E-8 &1.12E-8 \\
2\phantom{$\big\vert$} &55& 1.00E-3 & 8.00E-4 &4.67E-4 \\
2\phantom{$\big\vert$} &55& 1.00E-2 & 8.05E-3 &4.67E-3 \\
2\phantom{$\big\vert$} &55& 1.01E-1 & 8.15E-2 &4.80E-2 \\
\hline
4\phantom{$\big\vert$} &110& 1.00E-3 & 8.12E-4 &1.16E-3 \\
6\phantom{$\big\vert$} &165& 1.00E-3 & 8.25E-4 &2.31E-3 \\
8\phantom{$\big\vert$} &220& 1.00E-3&8.41E-4  &3.61E-3 \\
\hline
\end{tabular}}
\vspace{1ex}
\caption{\red{Numerical stability analysis for the experimental tests using different relative perturbations and mesh sizes. }}
\label{t:pertM}
\end{table}


\section{Conclusions}

This paper gives some simple rules that guarantee that the finite element inverse problem has a unique solution, without resorting to regularisation techniques. Briefly, from the numerical problems tested, the most practical results can be summarised as follows:
\begin{itemize}
\item Use a nodally interpolated traction field $T^h$.
\item Obtain as many tractions  degrees of freedom as observed displacements, i.e. impose $m=n_0$. This ensures a full-rank system and that the equilibrium equations are exactly satisfied.
\item If condition $m=n_0$ is not possible, use $m<n_0$ (less tractions than known displacements), but as a general rule, use $m>0.6n_0$ in order to avoid too large errors in the equilibrium equations.
\item Include unknown nodal displacements $\bsmu_1$, that are also computed through the inverse analysis. The higher the number of unknown displacements, the smaller the error in the equilibrium equations. However, in order to keep the condition number  $\kappa_I$ of the inverse problem not too large, and far below $\kappa_D^2$, with $\kappa_D$ the condition number of the direct problem, it is advised to limit the total number of displacement dofs $n$ according to the relation $m>0.05n$. 
\end{itemize}

We note that while the relation $m\le n_0$ is general, the conditions $0.6n_0<m$ and $0.05n<m$ have been obtained using cubic hexahedral elements, and thus may vary if other aspect ratios and geometries are employed.

Very often, the traction field is computed by resorting to the Boussinesq analytical solution for a linear material \cite{landau67}, and applying the Fourier transform of the solution, which yields a set of uncoupled system of equations. This technique can be applied to those situations where the Green function is known, like an infinite half-plane with infinite thickness \cite{butler02} or with a constant bounded thickness \cite{trepat09,alamo13}. In both cases, the material is assumed linear and homogeneous. The finite element approach presented here, and the results derived, may be also  applied to arbitrary non-homogeneous domains. 

An example of the use of FE techniques in TFM may be found for instance in \cite{legant13}. These references do not exploit the results shown in the present paper, and consequently regularisation was employed. In non-linear elasticity, the conclusions stated here do not necessarily carry over the resulting system of non-linear equations, which requires an iterative process \cite{palacio13}. 

Another common approach in TFM is the so-called direct forward method, which after interpolating the strain field, computes the tractions from the derived stresses as,
\[
\bt=\bssigma(\bu)\bn,
\]
with $\bn$ the external normal of the boundary. This approach may be employed in linear \cite{franck11,hur09,hall13} and non-linear elasticity \cite{toyjanova14}. However, in this method, the derived stress tensor $\bssigma$ does not necessarily satisfies the equilibrium condition $\nabla\cdot\bssigma=\0$ due to the assumed constitutive law of the material and experimental errors when measuring the displacement field. Instead, the traction field obtained from the finite element technique presented here minimises the error of the equilibrium equations, and when $m=n_0$, these are exactly satisfied  (in a weak sense). 
%


\section*{Acknowledgements}

The author is greatly thankful to Xavier Trepat and Vito Conte from the Institute of Bioenginyeria de Barcelona (IBEC), Spain, for their fruitful discussions. The grant DPI2013-43727-R, from the Spanish Ministry of Economy and Competitiveness (MinECo), is financially acknowledged.

\bibliographystyle{plain}

\begin{thebibliography}{10}

\bibitem{ainsworth00}
M.~Ainsworth and J.~T. Oden.
\newblock {\em {A posteriori error estimation in finite element analysis}}.
\newblock John Wiley \&\ Sons, New York, 1st edition, 2000.

\bibitem{ambrosi06}
D~Ambrosi.
\newblock Cellular traction as an inverse problem.
\newblock {\em SIAM J.\@ Appl.\@ Math.\@}, 66(6):2049--2060, 2006.

\bibitem{asadzadeh10}
M~Asadzadeh and L~Beilina.
\newblock A posteriori error estimates in a globally convergent {FEM} for a
  hyperbolic coefficient inverse problem.
\newblock {\em Inverse Problems}, 26(11):115007, 2010.

\bibitem{beilina14}
L~Beilina, N~T Th{\`a}nh, M~V Klibanov, and M~A Fiddy.
\newblock Reconstruction from blind experimental data for an inverse problem
  for a hyperbolic equation.
\newblock {\em Inverse Problems}, 30:025002, 2014.

\bibitem{brugues14}
A~Brugu{\'e}s, E~Anon, V~Conte, JH~Veldhuis, M~Gupta, J~Collombelli, J~J
  Mu{\~n}oz, GW~Brodland, B~Ladoux, and X~Trepat.
\newblock Forces driving epithelial wound healing.
\newblock {\em Nature Phys.\@}, 10:683--690, 2014.

\bibitem{butler02}
J.P. Butler, I.M. Toli{\'c}-N{\o}rrelykke, B.~Fabry, and J.J. Fredberg.
\newblock Traction field, moments, and strain energy that cells exert on their
  surroundings.
\newblock {\em Amer.\@ J.\@ Physiol.\@ Cell Physiol.\@}, 282:C595--C605, 2002.

\bibitem{ciarlet88}
PG~Ciarlet.
\newblock {\em Mathematical Elasticity}, volume~1 of {\em Studies in
  mathematics and its applications}.
\newblock Elsevier Science BV, Amsterdam, 1988.

\bibitem{alamo13}
JC~del \'Alamo, R~Meili, B~\'Alvarez-Gonz\'alez, B~Alonso-Latorre, E~Bastounis,
  R~Firtel, and JC~Lasheras.
\newblock Three-dimensional quantification of cellular traction forces and
  mechanosensing of thin substrata by {Fourier} traction force microscopy.
\newblock {\em PLOS ONE}, 8(9):e69850, 2013.

\bibitem{brask15}
J.B. Brask, G.~Singla-Buxarrais, M.~Uroz, R.~Vincent, and X.~Trepat.
\newblock Compressed sensing traction force microscopy.
\newblock {\em Acta Biomat.\@}, 26:286--294, 2015.

\bibitem{dembo96}
M~Dembo, T~Oliver, A~Ishihara, and K~Jacobson.
\newblock Imaging the traction stresses exerted by locomoting cells with the
  elastic substratum method.
\newblock {\em Bioph.\@ J.\@}, 70:2008--2022, April 1996.

\bibitem{dembo99}
M~Dembo and YL~Wang.
\newblock Stresses at the cell-to-substrate interface during locomotion of
  fibroblasts.
\newblock {\em Bioph.\@ J.\@}, 76:2307--2316, April 1999.

\bibitem{deng13}
X.~Deng, Y.~Zhao, and J.~Zou.
\newblock On linear finite elements for simultaneously recovering source
  location and intensity.
\newblock {\em Int.\@ J.\@ Num.\@ Anal.\@ Mod.\@}, 10(3):588--602, 2013.

\bibitem{roure05}
O~du~Roure, A~Saez, A~Buguin, RH~Austin, P~Chavrier, P~Siberzan, and B~Ladoux.
\newblock Force mapping in epithelial cell migration.
\newblock {\em Proc.\@ Nat.\@ Acad.\@ Sci.\@ USA}, 102(7):2390--2395, 2005.

\bibitem{franck11}
C~Franck, SA~Maskarinec, DA~Tirrell, and G~Ravichandran.
\newblock Three-dimensional traction force microscopy: a new tool for
  quantifying cell-matrix interactions.
\newblock {\em PLOS ONE}, 6(3):e17833, 2011.

\bibitem{galbraith97}
CG~Galbraith and MP~Sheetz.
\newblock A micromachined device provides a new bend on fibroblast traction
  forces.
\newblock {\em Proc.\@ Nat.\@ Acad.\@ Sci.\@ USA}, 94:9114--9118, 1997.

\bibitem{hall13}
M.S. Hall, R.~Long, X.~Feng, Y.L. Huang, C.Y. Hui, and M.~Wu.
\newblock Toward single cell traction microscopy within 3{D} collagen matrices.
\newblock {\em Exp.\@ Cell Res.\@}, 319:2396--2408, 2013.

\bibitem{hansen00}
PC~Hansen.
\newblock The {L}-curve and its use in the numerical treatment of inverse
  problems.
\newblock In {\em Computational Inverse Problems in Electrocardiology, ed. P.
  Johnston, Advances in Computational Bioengineering}, pages 119--142. WIT
  Press, 2000.

\bibitem{harris80}
A.K. Harris, P.~Wild, and D.~Stopak.
\newblock Silicone rubber substrata: a new wrinkle in the study of cell
  locomotion.
\newblock {\em Science}, 208:177--179, 1980.

\bibitem{hughes87}
T.J.R. Hughes.
\newblock {\em {The finite element method. Linear static and dynamic finite
  element analysis}}.
\newblock Prentice-Hall International Editions, 1987.

\bibitem{hur09}
SS~Hur, Y~Zhao, YS~Li, E~Botvinick, and S~Chien.
\newblock Live cells exert 3-dimensional traction forces on their substrata.
\newblock {\em Cell.\@ Mol.\@ Bioeng.\@}, 2(3):425--436, 2009.

\red{\bibitem{jin10}
B.~Jin and J.~Zou.
\newblock Numerical estimation of the {R}obin coefficient in a stationary
  diffusion equation.
\newblock {\em IMA J.\@ Num.\@ Anal.\@}, 30:677--701, 2010.}

\bibitem{landau67}
L~Landau and E~Lisfchitz.
\newblock {\em Th{\'e}orie de l'elasticit{\'e}}.
\newblock Mir, Moscow, mir edition, 1967.

\bibitem{legant13}
WR~Legant, CK~Choi, JS~Miller, L~Shao, L~Gao, E~Betzig, and CS~Chen.
\newblock Multidimensional traction force microscopy reveals out-of-plane
  rotational moments about focal adhesions.
\newblock {\em Proc.\@ Nat.\@ Acad.\@ Sci.\@ USA}, 110(3):881--886, 2013.

\red{\bibitem{li11b}
J.~Li, J.~Xie, and J.~Zou.
\newblock An adaptive finite element reconstruction of distributed fluxes.
\newblock {\em Inverse Problems}, 27(7):075009, 2011.}

\bibitem{lions68}
JL~Lions.
\newblock {\em Contr{\^o}le optimal de syst{\`e}mes gouvern{\'e}s par des
  {\'e}quations aux d{\'e}riv{\'e}es partielles}, volume XII.
\newblock Paris, Dunod: Gauthier-Villars, Paris, 1968.

\bibitem{palacio13}
J~Palacio, A~Jorge-Pe{\~n}as, A~Mu{\~n}oz-Barrutia, C~Ortiz de~Solorzano, {E
  de} Juan-Pardo, and JM~Gar{c\'ia}-Aznar.
\newblock {Numerical estimation of 3D mechanical forces exerted by cells on
  non-linear materials}.
\newblock {\em J.\@ Biomechanics}, 46(1):50--55, 2013.

\bibitem{ramm05}
AG~Ramm.
\newblock {\em Inverse problems: mathematical and analytical techniques with
  applications to engineering}.
\newblock Springer, New York, 2005.

\bibitem{salsa08}
S~Salsa.
\newblock {\em Partial Differential in Equations in Action}.
\newblock Springer-Verlag, Milano, Italy, 2008.

\bibitem{schwarz02}
US~Schwarz, NQ~Balaban, D~Riveline, A~Bershadsky, B~Geiger, and SA~Safran.
\newblock Calculation of forces at focal adhesions from elastic substrate data:
  the effect of localized force and the need for regularization.
\newblock {\em Bioph.\@ J.\@}, 83:1380--1394, 2002.

\bibitem{schwarz15}
U.S. Schwarz and J.R.D. Soine.
\newblock Traction force microscopy on soft elastic substrates: A guide to
  recent computational advances.
\newblock 1853(11):3095--3104, 2015.

\bibitem{toyjanova14}
J~Toyjanova, E~Bar-Kochba, C~L{\'o}pez-Fagundo, J~Reichner, D~Hoffman-Kim, and
  C~Franck.
\newblock High resolution, large deformation 3{D} traction force microscopy.
\newblock {\em PLOS ONE}, 9(4):e90976, 2014.

\bibitem{trepat09}
X.~Trepat, M.R. Wasserman, T.E. Angelini, E.~Millet, D.A. Weitz, J.P. Butler,
  and J.J. Fredberg.
\newblock {Physical forces during collective cell migration}.
\newblock {\em Nature Phys.\@}, 5(3):426--430, 2009.

\bibitem{vitale12}
G~Vitale, L~Preziosi, and D~Ambrosi.
\newblock Force traction microscopy: An inverse problem with pointwise
  observations.
\newblock {\em J.\@ Math.\@ Anal.\@ Appl.\@}, 395:788--801, 2012.

\bibitem{vitale12b}
G~Vitale, L~Preziosi, and D~Ambrosi.
\newblock A numerical method for the inverse problem of cell traction in 3d.
\newblock {\em Inverse Problems}, 28:095013, 2012.

\red{\bibitem{xu15}
Y.~Xu and J.~Zou.
\newblock Analysis of an adaptive finite element method for recovering the
  robin coefficient.
\newblock {\em SIAM J.\@ Control Optim.\@}, 53(2):622--644, 2015.}

\end{thebibliography}


\end{document}